\documentclass[10pt]{article}

\usepackage{amsmath}
\usepackage{amsfonts}
\usepackage{mathrsfs}
\usepackage{amsthm}
\usepackage{bm}
\usepackage{url}
\usepackage[T1]{fontenc}
\usepackage{authblk}
\usepackage{cite}
\usepackage{csquotes}
\MakeOuterQuote{"}

\theoremstyle{plain}
\newtheorem{theorem}{Theorem}
\newtheorem{lemma}{Lemma}

\newtheorem{corollary}{Corollary}

\theoremstyle{definition}

\theoremstyle{remark}
\newtheorem{remark}{Remark}

\newcommand{\mrm}[1]{\mathrm{#1}}
\newcommand{\tr}{\operatorname{tr}}

\newcommand{\sgn}{\operatorname{sgn}}
\newcommand{\ad}{\operatorname{ad}}

\newcommand{\rmT}{\mathrm{T}}
\newcommand{\imply}{\mathrel{\Rightarrow}}
\newcommand{\equi}{\mathrel{\Leftrightarrow}}

\newcommand{\be}{\begin{equation}}
\newcommand{\ee}{\end{equation}}
\newcommand{\ba}{\begin{align}}
\newcommand{\ea}{\end{align}}

\def\<{\langle}  
\def\>{\rangle}  
\newcommand{\ket}[1]{| #1\>}

\newcommand{\dket}[1]{| #1\>\!\>}

\newcommand{\dbra}[1]{\<\!\< #1|}

\def\outer#1#2{|#1\>\<#2|}       
\newcommand{\dinner}[2]{\<\!\< #1| #2\>\!\>}

\newcommand{\douter}[2]{| #1\>\!\>\<\!\< #2|}

\def\eqref#1{\textup{(\ref{#1})}}  
\newcommand{\eref}[1]{Eq.~\textup{(\ref{#1})}}
\newcommand{\Eref}[1]{Equation~\textup{(\ref{#1})}}
\newcommand{\esref}[1]{Eqs.~\textup{(\ref{#1})}}
\newcommand{\Esref}[1]{Equations~\textup{(\ref{#1})}}

\newcommand{\sref}[1]{Sec.~\ref{#1}}
\newcommand{\Sref}[1]{Section~\ref{#1}}

\newcommand{\thref}[1]{Theorem~\ref{#1}}
\newcommand{\Thref}[1]{Theorem~\ref{#1}}
\newcommand{\thsref}[1]{Theorems~\ref{#1}}

\newcommand{\lref}[1]{Lemma~\ref{#1}}

\newcommand{\crref}[1]{Corollary~\ref{#1}}

\newcommand{\cref}[1]{Conjecture~\ref{#1}}
\newcommand{\Cref}[1]{Conjecture~\ref{#1}}

\newcommand{\rcite}[1]{Ref.~\cite{#1}}
\newcommand{\rscite}[1]{Refs.~\cite{#1}}

\begin{document}
\title{Group theoretic, Lie algebraic and Jordan algebraic formulations of  the SIC existence problem}

\author{D. M. Appleby}
\affil{Centre for Engineered Quantum Systems, School of Physics,\newline The University of Sydney, Sydney, NSW 2006, Australia}

\author{Christopher A. Fuchs}
\affil{Quantum Information Processing Group, Raytheon BBN Technologies, \newline 10 Moulton Street, Cambridge, MA 02138, USA}

\author{Huangjun Zhu}
\affil{Perimeter Institute for Theoretical Physics, Waterloo, On N2L 2Y5, Canada \newline Electronic address: hzhu@pitp.ca}

\date{\today}
\maketitle

\begin{abstract}
Although symmetric informationally complete positive operator valued measures (SIC POVMs, or SICs for short)  have been constructed in every dimension up to 67, a general existence proof remains elusive.  The purpose of this paper is to show that the SIC existence problem is equivalent to  three other, on the face of it quite different problems.  Although it is still not clear whether these reformulations of the problem will make it more tractable,  we  believe that the fact that SICs have these connections to other areas of mathematics is of some intrinsic interest.  Specifically, we reformulate the SIC problem in terms of (1) Lie groups, (2) Lie algebras and (3) Jordan algebras (the second  result  being a  greatly strengthened version of one previously obtained by Appleby, Flammia and Fuchs).  The connection between these three reformulations is non-trivial:  It is not easy to demonstrate their equivalence directly, without appealing to their common equivalence to SIC existence.
In the course of our analysis we obtain a number of other results which may be of some independent interest.
\end{abstract}

\tableofcontents

\section{Introduction}

In a $d$-dimensional Hilbert space, a \emph{symmetric informationally complete positive operator valued measure} (SIC~POVM, or SIC in short)\footnote{Also known as  symmetric informationally complete probability operator measurement (SIC~POM) in the physics community.} is composed of $d^2$ subnormalized projectors onto pure states
$E_j=|\psi_j\rangle\langle\psi_j|/d$ with equal pairwise fidelity~\cite{Zaun11,ReneBSC04},
\begin{equation} \label{eq:SICinner}
|\langle\psi_j|\psi_k\rangle|^2=\frac{d\delta_{jk}+1}{d+1}.
\end{equation}
SICs  have many important properties, which are
rooted in this simple description. They are
simultaneously minimal 2-designs and  maximal sets of equiangular
lines \cite{Welc74,DelsGS75,Hogg82,Zaun11, ReneBSC04,Rene04the}.  They are
optimal for linear quantum state tomography
\cite{Rene04the,RehaEK04,Scot06,ZhuE11,PetzR12O,PetzR12E} and measurement-based quantum cloning~\cite{Scot06}. They are useful in quantum cryptography
\cite{FuchS03,Rene04the,Rene05,EnglKNC04,DurtKLL08}, quantum
fingerprinting~\cite{ScotWS07}, and signal processing~\cite{HowaCM06}.
They  play a crucial role in studying foundational issues
\cite{Fuch03,FuchS13,Fuch10,Rene04the} and in understanding the geometry
of quantum state space \cite{Beng05,BengZ06book,ApplEF11}. They have intriguing connections with mutually unbiased bases (MUB)
\cite{Ivan81,
WootF89,DurtEBZ10,Woot06,ApplDF07,AlboK07,Appl09S,KaleSE12E,KaleSE12S} and discrete
Wigner functions \cite{Rene04the,ColiCDG05,Appl07}.    They are very  interesting from a mathematical point of view, having connections with Galois theory~\cite{ApplAZ13}, Lie algebras~\cite{ApplFF11}, and the graph isomorphism problem~\cite{Zhu12the}.  They have attracted the attention of experimentalists.   Qubit SICs
\cite{DuSPD06,DurtKLL08, LingLK08} and qutrit SICs~\cite{MedeTST11}
have now been implemented in experiments.  There have appeared recent proposals for
realizing SICs  by successive measurements~\cite{KaleSE12E,KaleSE12S} and by multiport devices~\cite{Tabi12}.

Most studies on  SICs have assumed group
covariance \cite{Zaun11, ReneBSC04, Rene04the, ScotG10, Zhu12the}, partly because group covariant SICs are much
easier to construct and to analyze. In fact, all known SICs are group covariant, and almost all of them
are covariant with respect to the Heisenberg-Weyl group, also known as the generalized Pauli group \cite{Zaun11, ReneBSC04, Rene04the, ScotG10, Zhu10,Zhu12the}.
Up to now, analytical solutions of SICs and numerical solutions with high precision have been found up to dimension 67
\cite{DelsGS75, ReneBSC04,Zaun11,Gras04,Appl05,Gras05,Gras06,Gras08S,Gras08C, ScotG10,ApplBBG12,ApplAZ13}.
This encourages the belief that SICs exist in every finite dimensional Hilbert space.  However,  there is no universal recipe for constructing SICs despite the efforts of many researchers in the past decade.   Apart from a few low-dimensional cases where SICs have been obtained using "pencil and paper" methods \cite{Zaun11, Appl05, ApplBBG12}, most known solutions have been obtained either numerically, by minimizing the frame potential~\cite{ReneBSC04,ScotG10}, or analytically, by constructing a Gr\"{o}bner basis~\cite{Gras08C,ScotG10}. Both these methods are computationally very demanding, and the time for the calculation grows rapidly with the dimension.
Therefore, it is increasingly
difficult to obtain new solutions without introducing new ideas.

Besides the construction of SICs, a major open problem is the  SIC existence problem: Do SICs exist in every finite dimensional Hilbert space? This problem is crucial to understanding the geometry of quantum state space and to decoding quantum mechanics from a Bayesian point of view~\cite{FuchS13,ApplEF11}. The existence of SICs is also  equivalent to the existence of many interesting objects  appearing in various contexts: such as maximal sets of equiangular lines, minimal 2-designs \cite{Welc74,DelsGS75,Hogg82,Zaun11, ReneBSC04,Rene04the},  best approximation to orthonormal bases among bases composed of positive operators~\cite{ApplDF07}, minimal efficient tight informationally complete measurements~\cite{Scot06}, and minimal decomposition of certain separable states~\cite{Chen13}.
Therefore, any progress on the SIC existence problem will be beneficial to a wide range of subjects. Despite the simple description, however, this problem is  extremely difficult to attack directly. To make further progress, it is indispensable to introduce new lines of thinking.

The primary purpose of this paper is to describe three different reformulations of the SIC existence problem as  (1) a problem concerning Lie groups, (2) a problem concerning Lie algebras  and (3) a problem concerning Jordan algebras.  Surprisingly, although these  problems are all equivalent, it is not easy to establish this fact without using the link through SICs.  To ensure maximum generality we do not assume group covariance.   In the course of obtaining our main results we derive a number of other geometric, combinatoric and information theoretic results which may be of some independent interest.

The first of our main results concerns the real orthogonal group $\mathrm{O}(d^2)$.  We say that a subgroup $G \subseteq \mathrm{O}(d^2)$ is stochastic if its elements are all of the form
\begin{equation}
R = (d+1) S - d P,
\end{equation}
where $S$ is a doubly stochastic matrix~\cite{MarsOA11book,HornJ85book} (\emph{i.e.}\ a matrix whose matrix elements are non-negative and whose rows and columns sum to $1$) and $P$ is the fixed rank-$1$ projector
\begin{equation}\label{eq:AllOneMatrix}
P =\frac{1}{d^2} \begin{pmatrix} 1 & 1 & \dots & 1 \\ 1 & 1 & \dots & 1 \\ \vdots & \vdots &  & \vdots \\ 1 & 1 & \dots & 1\end{pmatrix}.
\end{equation}
We will show that a SIC exists in dimension~$d$ if and only if $\mathrm{O}(d^2)$ contains a stochastic subgroup isomorphic to the projective unitary group $\mrm{PU}(d)$ (\emph{i.e.}\ the unitary group in dimension~$d$  \emph{modulo} its center). We find this connection between orthogonal matrices and doubly stochastic
matrices surprising.  We discovered it while investigating the symmetry
properties of sets of probability distributions known as maximal consistent
sets.  In turn, these are motivated by Quantum Bayesianism~\cite{FuchS13,ApplEF11,ApplFZ13M}.  Closely related to this result we derive a bound on the matrix elements of the adjoint representation matrices of the unitary group in dimension~$d$, and we show that the inequality is saturated if and only if a SIC exists in dimension~$d$.  Besides their relevance to the SIC existence problem we believe these results may also be interesting to group theorists.

The second of our main results concerns the Lie algebra of the unitary group in dimension~$d$.  Let $L=\{L_j\}$ be a basis for the algebra, and  $C^L_j$ the adjoint representation matrices of the basis elements.  We will show that a SIC exists in dimension $d>2$ if and only if there exists a basis $L$ such that the $C^L_j$ are Hermitian and rank $2(d-1)$.  This result  greatly strengthens a result previously obtained by Appleby, Flammia and Fuchs~\cite{ApplFF11} (note, however, that the result proved here only holds for $d>2$, whereas the one  in \rcite{ApplFF11} holds for $d\ge 2$).

The third of our main results concerns the Jordan algebra consisting of all operators on the $d$-dimensional Hilbert space and equipped with the anti-commutator as product.  Let $L=\{L_j\}$ be a Hermitian basis for this algebra,  $C^L_{jkl}$ the structure constants and  $C^L_j$  the structure matrices defined by $(C^L_j)_{kl} = C^L_{jkl}$.
We  will show that a SIC exists in dimension $d>2$ if and only if there exists a basis $L$ such that each structure matrix is a  linear combination of a rank-$(2d-1)$ real symmetric matrix and the identity matrix.
 We also prove  a weaker version of this result which holds for $d\ge 2$.  This Jordan algebraic formulation of the SIC existence problem may be relevant to convex-operational approaches to quantum mechanics, given the close connections between Jordan algebras and homogeneous self-dual cones \cite{Wilc09,Wilc11,BarnW12}.

The rest of this paper is organized as follows.  In \sref{sec:simplex} we investigate a type of structure which could be described as a generalized $2$-design \cite{Zaun11,ReneBSC04, Scot06}.  This section establishes the basic  framework on which everything else in the paper depends.  In \sref {sec:GeoComInf} we begin by showing how the results in \sref{sec:simplex} can be used to give a simple, unified treatment of several well-known geometric, combinatoric and information theoretic  results.  We then go on to establish several technical results  needed in the sequel.  In \sref{sec:AdjointUnitary} we present two group theoretic formulations of the SIC existence problem.  In \sref{sec:LieAlgebra} we present  the Lie algebraic formulation of the SIC existence problem.  In \sref{sec:JordanAlgebra}, we present  the Jordan algebraic formulation of the SIC existence problem.  \Sref{sec:summary} summarizes the paper.

\section{\label{sec:simplex}Much ado about simplices}
The projectors $\Pi_j$ defining a SIC in dimension~$d$ form a $(d^2-1)$-dimensional regular simplex in the space of Hermitian operators:
\begin{equation}
\tr(\Pi_j \Pi_k) = \frac{d\delta_{jk} +1}{d+1}.
\end{equation}
They also form   a $2$-design \cite{Zaun11,ReneBSC04, Scot06},
\begin{equation}
\sum_j \Pi_j \otimes \Pi_j = \frac{2d}{d+1} P_{\mrm{s}},
\end{equation}
where $P_{\mrm{s}}$ is the projector onto the symmetric subspace of the tensor-product space.
In this section we consider  more general families of Hermitian operators which are not required to be either rank $1$ or positive, and establish connections between the simplices they define and what might be called  generalized $2$-designs.   In this way we derive several simple yet useful  results which will serve as a unified basis for
studying SICs from various perspectives, including but not restricted to geometric, combinatoric, algebraic, group theoretic, and information theoretic lines of thinking.
As we will see, this approach is surprisingly powerful both for rederiving old results and for obtaining new ones.

Throughout the rest of the paper, $\mathcal{H}$ denotes a $d$-dimensional Hilbert space and  $\mathcal{B}(\mathcal{H})$  the space of operators on $\mathcal{H}$ with  identity 1. The space  $\mathcal{B}(\mathcal{H})$ is itself a Hilbert space equipped with Hilbert-Schmidt inner product $\dinner{A}{B}:=\tr(A^\dag B)$ for $A,B\in \mathcal{B}(\mathcal{H})$, where we have used double ket notation to distinguish operator kets from ordinary ones~\cite{ZhuE11,Zhu12the}. Superoperators, such as the
outer product $\douter{A}{A}$, act on this space just as
operators  on the ordinary Hilbert space; the identity superoperator is denoted by $\mathbf{I}$ (the arithmetic of
superoperators can be found in \rscite{RungMND00, RungBCH01, DariPS00, ZhuE11,Zhu12the}).
$P_\mrm{s}$ and $P_\mrm{a}$ denote the projectors onto the symmetric and anti-symmetric subspaces, respectively, of $\mathcal{H}^{\otimes2}$.

\begin{theorem}\label{thm:GeoComInf}
Suppose $\{L_j \}$  is a basis for  $\mathcal{B}(\mathcal{H})$ consisting of Hermitian operators. Then the following equations  are equivalent:
\begin{align}
\tr(L_jL_k)&=\alpha\delta_{jk}+\gamma\tr(L_j)\tr(L_k) \label{eq:Ginner},  \\
\sum_j L_j\otimes L_j&=(\beta+\alpha)P_{\mrm{s}}+(\beta-\alpha)P_{\mrm{a}}, \label{eq:G2design}\\
\sum_j \douter{L_j}{L_j}&=\alpha\mathbf{I}+\beta\douter{1}{1}.  \label{eq:GframeSuper}
\end{align}
In that case,
\begin{align}
\alpha & > 0, &  \alpha+ d\beta & > 0, & \gamma=\frac{\beta}{\alpha+d\beta}.
\end{align}
\end{theorem}

\begin{remark}
  \Eref{eq:Ginner} characterizes the  geometrical properties of the (possibly irregular) simplex formed by the vectors $\{L_j\}$ while \eref{eq:G2design} is what we are calling the generalized $2$-design property.  The theorem thus generalizes the connection between simplices and $2$-designs which we see in the case of a SIC. \Eref{eq:G2design}  reflects the combinatoric properties of $\{L_j\}$.  It is also relevant to the study of entanglement and minimal decomposition of separable states~\cite{Chen13}. \Eref{eq:GframeSuper} has information theoretic content: When the $L_j$ form a generalized measurement and have the same trace of $1/d$, the superoperator $\sum_j\douter{L_j}{L_j}$ determines  the efficiency of this measurement in linear state tomography \cite{Scot06, ZhuE11, Zhu12the} (see also \sref{sec:TightIC}). Remarkably, the geometric, combinatoric, and information theoretic aspects of the basis $\{L_j\}$  can be connected by a simple theorem. As we shall see shortly, quite a few key results pertinent  to SICs can be derived or rederived based on this theorem. Furthermore, \thref{thm:GeoComInf} is also the cornerstone for establishing group theoretic and algebraic formulations of the SIC existence problem, which are the main focus of this paper.
\end{remark}
The proof of the theorem depends on the following lemma, which is also  of some independent interest.
\begin{lemma}\label{lem:ComInf}
Suppose $\{L_j\}$  is a  set of $n$ Hermitian operators in $\mathcal{B}(\mathcal{H})$. Then the following two  statements are equivalent:
\begin{align}
\sum_j L_j\otimes L_j&=(\beta+\alpha)P_{\mrm{s}}+(\beta-\alpha)P_{\mrm{a}},\label{eq:GG2design}\\
\sum_j \douter{L_j}{L_j}&=\alpha\mathbf{I}+\beta\douter{1}{1}.   \label{eq:GGframeSuper}
\end{align}
The set $\{L_j\}$ spans $\mathcal{B}(\mathcal{H})$ if and only if
$\alpha > 0$ and $\alpha+ d\beta > 0$.
\end{lemma}
\begin{remark}
Note that this result applies more generally than  \thref{thm:GeoComInf} since we do not assume that  $\{L_j\}$ is a basis, nor even that it has cardinality $d^2$.
\end{remark}
\begin{proof}
Let $\mathcal{B}(\mathcal{B}(\mathcal{H}))$ denote the space of operators on $\mathcal{B}(\mathcal{H})$. Then $\mathcal{B}(\mathcal{B}(\mathcal{H}))$  is isometrically (with respect to the Hilbert-Schmidt inner product) isomorphic to $\mathcal{B}(\mathcal{H})\otimes\mathcal{B}(\mathcal{H})$  under the map $\douter{A}{B} \rightarrow  A\otimes B^\dag$ for $A,B\in \mathcal{B}(\mathcal{H})$. The equivalence of \esref{eq:GG2design} and \eqref{eq:GGframeSuper} follows from the observation that under this isomorphism $\douter{1}{1}$ is mapped to the identity, which is equal to $P_{\mrm{s}}+P_{\mrm{a}}$, and $\mathbf{I}$ is mapped to the swap operator, which is equal to $P_{\mrm{s}}-P_{\mrm{a}}$.  Here the first claim follows from the definition. To verify the latter claim, define $E_{rs}=\outer{r}{s}$. Then $\mathbf{I}=\sum_{r,s} \douter{E_{rs}}{E_{rs}}$, whose image under the isomorphism is exactly the swap operator $\sum_{r,s} \outer{rs}{sr}$.

Finally, note that the operators $L_j$ span  $\mathcal{B}(\mathcal{H})$ if and only if the superoperator $\sum_j \douter{L_j}{L_j}$ is positive definite, which is true if and only if its two distinct eigenvalues $\alpha$ and $\alpha + d\beta$ are both positive.
\end{proof}

For later reference let us note that it is easy to obtain explicit expressions for the constants $\alpha, \beta$ featuring in the lemma.  In fact, taking the trace of both sides of   \eref{eq:GGframeSuper} gives
\begin{align}\label{eq:LjSumA}
\sum_j \tr(L_j^2)&=d^2\alpha +d\beta,
\\
\intertext{while taking the  inner product  with $\douter{1}{1}$ gives}
\sum_j[\tr(L_j)]^2&=d\alpha +d^2\beta.\label{eq:LjSumB}
\end{align}
Consequently,
\begin{equation}
\alpha=\frac{d\sum_j\tr(L_j^2)-\sum_j[\tr(L_j)]^2}{d^3-d},  \quad  \beta=\frac{-\sum_j\tr(L_j^2)+d\sum_j[\tr(L_j)]^2}{d^3-d}. \label{eq:alphabeta}
\end{equation}
Observe also that multiplying both sides of \eref{eq:GGframeSuper} by $\dket{1}$ on the right gives
\begin{equation}
\sum_j \tr(L_j)L_j=\alpha+d\beta. \label{eq:LjLj}
\end{equation}

\begin{proof}[Proof of \thref{thm:GeoComInf}]
The equivalence of \esref{eq:G2design} and \eqref{eq:GframeSuper} follows from \lref{lem:ComInf}.  To show that \eref{eq:GframeSuper} implies \eref{eq:Ginner},
denote the Gram matrix of $\{L_j\}$ by $M$. Then $M$ and $\sum_j\douter{L_j}{L_j}$ have the same spectrum. If \eref{eq:GframeSuper} holds, then all eigenvalues of $\douter{L_j}{L_j}$ or, equivalently, of $M$ are equal to $\alpha$ except for one equal to $\alpha+d\beta$. It follows from \eref{eq:LjLj} that
\begin{equation}
\sum_j \tr(L_j)\tr(L_jL_k)=(\alpha+d\beta)\tr(L_k),
\end{equation}
which implies that the vector $(\tr(L_1),\tr(L_2), \ldots,\tr(L_{d^2}))^\rmT$ is the eigenvector of $M$ with eigenvalue $\alpha+d\beta$. As a consequence,
\begin{equation}
\tr(L_jL_k)=\alpha\delta_{jk}+\frac{d\beta\tr(L_j)\tr(L_k)}{\sum_j [\tr(L_j)]^2}=\alpha\delta_{jk}+\frac{\beta}{\alpha+d\beta}
\tr(L_j)\tr(L_k),
\end{equation}
where in deriving the last equality we have applied \eref{eq:LjSumB}.  Note that \lref{lem:ComInf} guarantees that $\alpha+d\beta$ is non-zero since $\{L_j\}$ is a basis.
So \eref{eq:Ginner} holds with $\gamma = \beta/(\alpha+d\beta)$.

It remains to show that  \eref{eq:Ginner}  implies \eref{eq:GframeSuper}.   If \eref{eq:Ginner}  holds,  then
\begin{equation}
\tr \biggl[\biggl(\sum_j \tr(L_j)L_j \biggr) L_k\biggr]=\biggl\{\alpha+\gamma\sum_j[\tr(L_j)]^2\biggr\}\tr(L_k),
\end{equation}
which implies  that
\begin{equation}\label{eq:thm1SubsidB}
\sum_j\tr(L_j)L_j=\alpha+\gamma\sum_j[\tr(L_j)]^2
\end{equation}
since $\{L_j\}$ is a basis in the operator space. Taking the trace on both sides we find
\begin{equation}\label{eq:thm1SubsidA}
\sum_j [\tr(L_j)]^2 = d\alpha +d \gamma \sum_j[\tr{L_j)}^2.
\end{equation}
Now the fact that $\{L_j\}$ is a basis means that the Gram matrix $M_{jk}=\tr(L_jL_k)$ must be positive definite, implying that $\alpha>0$ and, consequently, $d\gamma \neq 1$ in view of \eref{eq:thm1SubsidA}.  We may therefore rearrange the equation to derive
\begin{equation}\label{eq:gammaEq}
\sum_j[\tr(L_j)]^2=\frac{d\alpha}{1-d\gamma}.
\end{equation}
\Esref{eq:thm1SubsidB} and \eqref{eq:gammaEq} imply that $\dket{1}$ is an eigenvector of the superoperator  $\sum_j\douter{L_j}{L_j}$ with eigenvalue $\alpha+[d\alpha\gamma/(1-d\gamma)]$.
 Now \eref{eq:GframeSuper} with $\beta=\alpha\gamma/(1-d\gamma)$ follows from the observation that all eigenvalues of $M$, that is, of  $\sum_j \douter{L_j}{L_j}$  are equal to $\alpha$ except for one equal to $\alpha+[d\alpha\gamma/(1-d\gamma)]$.
\end{proof}

From  \eref{eq:thm1SubsidA} or \eqref{eq:gammaEq} and the fact that $\alpha > 0$ we find
\begin{equation}\label{eq:gamma}
 \quad \gamma=\frac{1}{d}-\frac{\alpha}{\sum_j[\tr(L_j)]^2}<\frac{1}{d}.
 \end{equation}

An important special case of \thref{thm:GeoComInf} is  when $\alpha=\beta=d/(d+1)$:
\begin{corollary}\label{cor:GeoComInfSpecial}
Suppose $\{L_j\}$ is a set of $d^2$ Hermitian operators in $\mathcal{B}(\mathcal{H})$. Then the following three  equations are equivalent:
\begin{align}
\tr(L_jL_k)&=\frac{1}{d+1}\bigl[d\delta_{jk}+\tr(L_j)\tr(L_k)\bigr], \label{eq:gequi}\\
\sum_j L_j\otimes L_j&=\frac{2d}{d+1}P_{\mrm{s}}, \label{eq:g2design}\\
\sum_j \douter{L_j}{L_j}&=\frac{d}{d+1}(\mathbf{I}+\douter{1}{1}). \label{eq:superopeEq}
\end{align}
\end{corollary}
\begin{remark}Any equation in the corollary ensures that $\{L_j\}$ is a basis for $\mathcal{B}(\mathcal{H})$. If the $L_j$ were rank-$1$ projectors satisfying these conditions, they would define a SIC.
\end{remark}
\begin{proof}
The claim is a straightforward consequence of \thref{thm:GeoComInf}.\end{proof}

Another important special case of \thref{thm:GeoComInf} is when the basis  $\{L_j\}$  forms a regular simplex.  We conclude this section with two corollaries concerning this case, which will be needed in the sequel.  Since they are of a somewhat technical character we relegate the proofs to the appendix.
\begin{corollary}\label{cor:GeoComInf}
Suppose $\{L_j\}$ is a basis of Hermitian operators for $\mathcal{B}(H)$ which satisfies any of the three equivalent \esref{eq:Ginner}, \eqref{eq:G2design}, and \eqref{eq:GframeSuper} in \thref{thm:GeoComInf}.
If $\beta\neq0$, then the following statements are equivalent:
\begin{enumerate}

\item  \label{it:1}  The value of $|\tr (L_j)|$ is independent of $j$.

\item  \label{it:2} The value of  $\tr (L_j^2)$ is independent of $j$.

\item  \label{it:3} The value of $[\tr (L_j)]^2/\tr (L_j^2)$ is independent of $j$.

\item  \label{it:4} The value of $d\tr (L_j^2)-[\tr (L_j)]^2$ is independent of $j$.

\item \label{it:5} $\tr(L_j)\neq 0$ and $\tr(L_jL_k)=\alpha\delta_{jk}+\beta\epsilon_j\epsilon_k/d$ for all $j, k$.

\item \label{it:6} $\tr(L_j)\neq 0$ and the vectors $\epsilon_jL_j$ are equiangular.

\item  \label{it:7} $\tr(L_j)\neq 0$ and $\sum_j \epsilon_j L_j$ is proportional to the identity.

\item \label{it:8} $\sum_j |\tr(L_j)|=d\sqrt{d\alpha+d^2\beta}$.
\end{enumerate}
Here $\epsilon_j$ is the sign of $\tr(L_j)$.
If these statements hold, then
\begin{equation}\label{eq:Lj}
\tr(L_j^2)=\frac{d\alpha+\beta}{d},\quad    \tr(L_j)=\epsilon_j\sqrt{\frac{\alpha+d\beta}{d}},\quad \sum_j\epsilon_j L_j =\sqrt{d(\alpha+d\beta)}.
\end{equation}
If $\beta=0$ statement  \ref{it:2} is automatic.  Statements   \ref{it:1}, \ref{it:3}, \ref{it:4}, \ref{it:7} and \ref{it:8} are  equivalent and imply statements \ref{it:5} and \ref{it:6}, which are also equivalent.     \Eref{eq:Lj} is still applicable when  statements \ref{it:1}, \ref{it:3},  \ref{it:4}, \ref{it:7} and \ref{it:8} hold.
\end{corollary}
\begin{remark}
If $\beta = 0$ then $\{L_j\}$ is automatically a regular simplex.
If $\beta \neq 0$ then $\{L_j\}$ is a regular simplex if and only if   statements \ref{it:1} to \ref{it:8} hold with $\epsilon_j=1$  for all $j$ or $\epsilon_j=-1$ for all $j$.  Note that $\alpha > 0$ since $\{L_j\}$ is a basis.
\end{remark}
\begin{proof}  See the appendix.   \end{proof}

\begin{corollary}\label{cor:RegularSimplexGen}
Suppose $\{L_j\}$  is a basis of Hermitian operators for $\mathcal{B}(H)$  which satisfies
\begin{equation}
\tr(L_jL_k)=\alpha\delta_{jk}+\zeta \quad \forall j, k.    \label{eq:RegularSimplex}
\end{equation}
Then $\alpha>0$ and  $\zeta>-\alpha/d^2$.  Moreover, the following statements are equivalent:
\begin{enumerate}

\item $\sum_j L_j $ is proportional to the identity.

\item  The value of $\tr (L_j)$ is independent of $j$.

\item $|\sum_j \tr(L_j)|=d\sqrt{d\alpha+d^3\zeta}$.

\end{enumerate}
If any of these statements holds, then
\begin{equation}\label{eq:RegularSimplexLj}
\tr(L_j)=\epsilon\frac{\sqrt{d\alpha+d^3\zeta}}{d},\quad \sum _j L_j=\epsilon\sqrt{d\alpha+d^3\zeta}, \quad \epsilon=\pm1.
\end{equation}
If in addition $\zeta\neq0$, then any of statements 1, 2, 3 holds if and only if $\{L_j\}$ satisfies  \thref{thm:GeoComInf} with the same $\alpha$ and $\gamma=d\zeta/(\alpha+d^2\zeta)$.
\end{corollary}

\begin{remark}
 \Eref{eq:RegularSimplex}  is a necessary and sufficient condition for  the basis  to be a regular simplex.
Geometrically, all the conditions in the corollary amount to the requirement that the center of that simplex is proportional to the identity.
\end{remark}
\begin{proof}  See the appendix. \end{proof}

\section{\label{sec:GeoComInf}Geometric, combinatoric, and information theoretic characterizations of SICs}
The purpose of this section is two-fold.  The propositions proved in the last section are surprisingly powerful, given the simplicity of the underlying geometrical intuition.  One aim of this section is to illustrate that power by rederiving a number of well-known results.  The other aim is to derive some new results which will be needed in the sequel.

The study of SICs has drawn much inspiration from the study of equiangular lines \cite{Haan48,LintS66,LemmS73,DelsGS75,Roy06the,GodsR09,Khat08the}, spherical
codes and designs \cite{DelsGS77,Hogg82,Hogg92,BannB09,ReneBSC04,Rene04the,Scot06,Kim07}, as
well as  frame theory \cite{DuffS52,Casa00,ReneBSC04,Rene04the,Scot06}. Recently, it has also found interesting connections with quantum state estimation \cite{RehaEK04, HayaHH05, Scot06, ZhuE11, Zhu12the, PetzR12E}, entanglement theory \cite{ZhuTE10t, ChenZW11, Chen13}, and Lie algebras (\rcite{ApplFF11} and this paper). In this section we  review several important geometric, combinatoric, and information theoretic characterizations of SICs originating from these studies. Specifically, using the unified approach introduced in \sref{sec:simplex}, we provide self-contained proofs of the well known results that SICs are maximal sets of equiangular lines \cite{Welc74, DelsGS77, Zaun11}, minimal 2-designs \cite{Zaun11,ReneBSC04, Scot06}, and minimal efficient tight informationally complete (IC) measurements \cite{Scot06}, and vice versa. We then generalize these results  in preparation for group theoretic and algebraic treatments of the SIC existence problem.

In the rest of the paper by a "SIC" we mean a set of $d^2$ pure projectors in dimension~$d$ with equal pairwise fidelity of $1/(d+1)$, rather than the POVM obtained by scaling the projectors by a factor of $1/d$.

\subsection{Maximal equiangular lines}
In the mathematical community, SICs have been studied under the
name of \emph{equiangular  lines} for more than half a century
\cite{Haan48,LintS66,LemmS73,DelsGS75,Roy06the,GodsR09,Khat08the}; see
\rcite{Khat08the} for a historical  survey. When the lines are
represented by pure states, the equiangular condition means that the
pairwise fidelities among the states are the same. A cursory
inspection of the Gram matrix of the lines reveals that there are at
most $d^2$ equiangular lines in a (complex) Hilbert space of dimension~$d$~\cite{DelsGS75}.
When the pairwise fidelity $\mu$
 is smaller than $1/(d+1)$, there is a
tighter bound for the number $n$ of lines,
\begin{equation}\label{eq:WelchBound}
n\leq\frac{d-\mu d}{1-\mu d},
\end{equation}
which is known as the \emph{Welch bound}~\cite{Welc74}. A set of equiangular lines
is \emph{tight} if it saturates the Welch bound. SICs stand out as sets of equiangular lines that
saturate both the Welch bound and the absolute upper bound.

\begin{theorem}\label{thm:MaximalEAL}
Suppose $\Pi_j$ are $n$ pure states in dimension~$d$ that form equiangular lines, that is, $\tr(\Pi_j\Pi_k)=\alpha\delta_{jk}+1-\alpha$ with $0<\alpha\leq1$. Then $n\leq d^2$ and  the upper bound is saturated if and only  if $\{\Pi_j\}$ is a SIC.
\end{theorem}
\begin{remark}
Here we do not assume that $\{\Pi_j\}$ forms tight equiangular lines, so it is not so obvious that it is a SIC when the absolute upper bound $d^2$ is saturated.
\end{remark}
\begin{proof}
The rank of the Gram matrix of $\{\Pi_j\}$ is equal to $n$ and is also equal to the rank of the superoperator $\sum_j\douter{\Pi_j}{\Pi_j}$, which is bounded from above by $d^2$. It follows that $n\leq d^2$. If $n=d^2$, then $\alpha=d/(d+1)$ according to \thref{thm:GeoComInf} and \eref{eq:alphabeta} with $L_j=\Pi_j$, given  that $\tr(\Pi_j)=\tr(\Pi_j^2)=1$. Therefore, $\{\Pi_j\}$ is a SIC.
\end{proof}

\begin{corollary}\label{cor:MaxSimple}
Suppose $\Pi_j\in \mathcal{B}(H)$ are $d^2$ positive operators with unit length that are equiangular among each other, that is, $\tr(\Pi_j\Pi_k)=\alpha\delta_{jk}+1-\alpha$. Then $\alpha\leq d/(d+1)$ and the upper bound is saturated if and only if $\{\Pi_j\}$ is a SIC.
\end{corollary}
\begin{remark}
Here we do not assume that   $\Pi_j$ have the same trace or they form a generalized measurement up to a scale factor, but these requirements are automatically satisfied when the upper bound is saturated. This  corollary shows that in a sense SICs are maximal simplices that can fit into the state space. This result is consistent with the observation in  \rcite{ApplDF07} that SICs are the best approximation to orthonormal bases among bases composed of positive operators.
\end{remark}

\begin{proof}
The inequality $\alpha\leq d/(d+1)$ follows from the equation
\begin{equation}
d^2\alpha+d^4(1-\alpha)=\tr\biggl[\biggl(\sum_j \Pi_j\biggr)^2\biggr]\geq  \frac{[\tr(\sum_j \Pi_j)]^2}{d}\geq \frac{\bigl[\sum_j\sqrt{\tr(\Pi_j^2)} \bigr]^2}{d}=d^3.
\end{equation}
 Here the second inequality is saturated if and only if all $\Pi_j$ have rank one. In that case, $\{\Pi_j\}$ is a SIC according to \thref{thm:MaximalEAL} (assuming $\alpha\neq0$), so $\sum_j\Pi_j=d$ and the first inequality is saturated automatically.
\end{proof}

\subsection{Minimal 2-designs}
Consider a weighted
set of states $\{|\psi_j\rangle, w_j\}$ with $0<w_j\leq 1$ and
$\sum_j w_j=d$. Given a positive integer $t$,  the order-$t$
\emph{frame potential} $\Phi_t$~\cite{ReneBSC04,Scot06} is defined
as
\begin{equation}
\label{eq:t-design} \Phi_t=\sum_{j,k}w_jw_k|\langle
\psi_j|\psi_k\rangle|^{2t}=\tr(B_t^2),\quad
B_t=\sum_{j}w_j(|\psi_j\rangle\langle \psi_j|)^{\otimes t}.
\end{equation}
Note that $B_t$ is supported on the $t$-partite symmetric subspace,
whose dimension is $\binom{d+t-1}{t}$. The frame potential $\Phi_t$
is bounded from below by $d^2\binom{d+t-1}{t}^{-1}$, and the bound
is saturated if and only if $B_t=d\binom{d+t-1}{t}^{-1}P_t$, where
$P_t$ is the projector onto the $t$-partite symmetric subspace.  The
weighted set $\{|\psi_j\rangle,w_j\}$ is  a (complex-projective)
\emph{weighted $t$-design} if the lower bound is saturated; it is a
$t$-design\index{$t$-design} if, in addition, all the weights $w_j$
are equal \cite{Hogg82,Hogg92,ReneBSC04, Scot06,RoyS07}. It follows from the  definition that a
weighted $t$-design is also a weighted $t^\prime$-design for
$t^\prime<t$.

For any pair of positive integers $d$ and $t$, there exists a
(weighted) $t$-design with a finite number of elements~\cite{SeymZ84}. The number  is bounded from below by 
\begin{equation}\label{eq:tdesignLBound} \binom{d+\lceil
t/2\rceil-1}{\lceil t/2\rceil}\binom{d+\lfloor t/2\rfloor-1}{
\lfloor t/2\rfloor},
\end{equation}
where $\lceil
t/2\rceil$ denotes the smallest integer not smaller than $t/2$, and $\lfloor t/2\rfloor$ the largest integer not larger than $t/2$ \cite{Hogg82, Leve98, Scot06}.
The bound is equal to $d,d^2,d^2(d+1)/2$ for $t=1,2,3$, respectively.
Any resolution of the identity consisting of pure states is a
weighted 1-design. SICs \cite{Zaun11,
ReneBSC04, Scot06, ScotG10} and complete sets of
MUB \cite{Ivan81,
WootF89, DurtEBZ10} are prominent examples of 2-designs.

Here we are mainly  interested in weighted
2-designs and their connection with SICs \cite{ReneBSC04,Rene04the, Scot06, RoyS07, Kim07, Zaun11, ZhuE11, Zhu12the}. In particular, we rederive the result of Scott that any minimal weighted 2-design is a SIC~\cite{Scot06}.
\begin{theorem}\label{thm:minimal2design}
Suppose $\{ \ket{\psi_j}, w_j>0\}$ is  a weighted 2-design with $n$ elements in dimension~$d$, that is,
\begin{equation}\label{eq:minimal2design}
\sum_j w_j (\outer{\psi_j}{\psi_j})^{\otimes 2}=\frac{2}{d+1}P_\mrm{s}.
\end{equation}
Then $n\geq d^2$ and  the lower bound is saturated if and only if
 $w_j=1/d$ and $\{\outer{\psi_j}{\psi_j} \}$ is a SIC.
\end{theorem}
\begin{proof}
Let $L_j=\sqrt{w_j}\outer{\psi_j}{\psi_j}$, then $\tr(L_j^2)=[\tr(L_j)]^2$. If \eref{eq:minimal2design} holds, then $\sum_j\douter{L_j}{L_j}=(\mathbf{I}+\douter{1}{1})/(d+1)$ by the isomorphism between   $\mathcal{B}(\mathcal{H})\otimes\mathcal{B}(\mathcal{H})$ and $\mathcal{B}(\mathcal{B}(\mathcal{H}))$  (see \lref{lem:ComInf}). Therefore, $\sum_j\douter{L_j}{L_j}$ has rank $d^2$ and $n\geq d^2$.
When $n=d^2$, \crref{cor:GeoComInf} (see items~3 and~5) with $\alpha=\beta=1/(d+1)$ applied to $\{L_j\}$ yields
\begin{equation}
\tr(L_jL_k)=\frac{d\delta_{jk}+1}{d(d+1)}.
\end{equation}
Therefore, $w_j=1/d$ and $\{\outer{\psi_j}{\psi_j}\}$ is a SIC.
\end{proof}

\begin{corollary}\label{cor:bound}
Suppose $L_j\in \mathcal{B}(\mathcal{H})$ are $d^2$ positive operators that satisfy the three equivalent equations in \thref{thm:GeoComInf}. Then $\beta\geq \alpha$ and
$1/(d+1)\leq \gamma<1/d $. The  lower bounds for $\beta$ and~$\gamma$ are saturated if and only if $\{\sqrt{d/(d+1)\alpha} L_j\}$ is a SIC.
\end{corollary}
\begin{proof}
The inequality $\beta\geq \alpha$ follows from \eref{eq:G2design} and the observation that $\sum_j L_j\otimes L_j$ is positive semidefinite. If  $\beta=\alpha$,
then $\sum_j L_j\otimes L_j=2\alpha P_\mrm{s}$, so $L_j\otimes L_j$ are supported on the symmetric subspace, which implies that  $L_j$ have rank one. According to \thref{thm:minimal2design}, $\{\sqrt{d/(d+1)\alpha} L_j\}$ is a SIC. The upper bound for $\gamma$ follows from \eref{eq:gamma}; the lower bound and the equality condition follow from the equality $\gamma=\beta/(\alpha+d\beta)$. Alternatively, the two bounds for $\gamma$ can be established by virtue of \esref{eq:gamma} and \eqref{eq:alphabeta}.
\end{proof}

\Thref{thm:minimal2design} and \crref{cor:bound} have an interesting consequence in entanglement theory as observed by Chen \cite{Chen13}: If the bipartite state $2P_\mrm{s}/d(d+1)$ can be written as a convex combination of $n$ product states $\sum_j w_j \rho_j\otimes \rho_j^\prime$, then $n\geq d^2$, and the lower bound is saturated if and only if $w_j=1/d^2$, $\rho_j^\prime=\rho_j$,  and $\{\rho_j\}$ is a SIC.

\subsection{\label{sec:TightIC}Tight informationally complete measurements}
A generalized measurement $\{ E_j\}$ is a \emph{tight IC measurement}~\cite{Scot06} if
\begin{equation}\label{eq:TightIC}
\mathcal{F}:=d\sum_j\frac{\douter{ E_j}{ E_j}}{\tr( E_j)}=\alpha \mathbf{I}+\beta\douter{1}{1}
\end{equation}
for some positive constants $\alpha,\beta$. In linear state tomography, one  needs to invert the frame superoperator $\mathcal{F}$ to compute the reconstruction operators, which is generally complicated. Tight IC measurements are characterized by particular simple frame superoperators and thus easy state reconstruction \cite{Scot06, ZhuE11, Zhu12the}.

Multiplying \eref{eq:TightIC} by $\dket{1}$ on the right gives $\alpha+d\beta=d$. Taking the trace of the equation yields
\begin{equation}
d^2\alpha+d\beta=d\sum_j\frac{\tr( E_j^2)}{\tr( E_j)}\leq d\sum_j\tr( E_j)=d^2,
\end{equation}
which implies that $\alpha\leq d/(d+1)$, and the inequality is saturated if and if all $ E_j$ have rank one. In linear state tomography with tight IC measurements, the resource required to reach a given precision is roughly inversely proportional to $\alpha$ \cite{Scot06, ZhuE11, Zhu12the}.
Therefore, a tight IC measurement with $\alpha=d/(d+1)$ is called efficient.

If the upper bound $\alpha=d/(d+1)$ is saturated and $ E_j=w_j\outer{\psi_j}{\psi_j}$, then $\{\ket{\psi_j}, w_j\}$ satisfies \eref{eq:minimal2design}
according to \eref{eq:TightIC} and \crref{cor:GeoComInfSpecial} and  is thus a weighted 2-design. Conversely, every weighted 2-design defines an efficient tight IC measurement. Now application of  \thref{thm:minimal2design} reproduces a well-known result of Scott~\cite{Scot06}.
\begin{theorem}
A rank-1 measurement $\{ E_j=w_j\outer{\psi_j}{\psi_j}\}$ with $n$ elements is a tight IC measurement if and only if $\{\ket{\psi_j}, w_j\}$ is a weighted 2-design.   It is a minimal tight IC measurement \textup{(}that is $n=d^2$\textup{)} if and only if $w_j=1/d$ and $\{\outer{\psi_j}{\psi_j} \}$ is a SIC.  A measurement with $d^2$ outcomes is an efficient tight IC measurement if and only if it is a SIC measurement.
\end{theorem}

\subsection{Generalization}
In this section we present a  result which will be needed in the sequel. It concerns   a special case of \thref{thm:GeoComInf} which  generalizes the connection between weighted 2-designs and SICs discussed above. Measurements of this form  are often used to model real experiments when there is white noise~\cite{ZhuE11}. Since the result is somewhat technical we relegate the proof to the appendix.

\begin{theorem}\label{thm:SIC}
Suppose  $\{L_j\} $ is a set of  $d^2$  Hermitian operators in $\mathcal{B}(\mathcal{H})$  satisfying the two equivalent \esref{eq:GG2design} and \eqref{eq:GGframeSuper} of \lref{lem:ComInf} with $\alpha > 0$.  Suppose also that each $L_j$ is a linear combination of a rank-1 projector and the identity. If $d\geq3$,  then there exists a SIC  $\{\Pi_{j}\}$ such that  $L_j = a_j \Pi_j + b_j$ with
\begin{equation}\label{eq:ajbj}
\quad a_j = \epsilon_j \sqrt{\frac{\alpha(d+1)}{d}},   \quad b_j = -\frac{a_j}{d} \left( 1 - \epsilon \sqrt{\frac{\alpha+d\beta}{\alpha(d+1)}}\right),
\end{equation}
where  the $\epsilon_j$ are  signs and $\epsilon$ is a fixed sign.

The same conclusion holds when $d=2$  if in addition $\{L_j\} $ is a basis  for $\mathcal{B}(\mathcal{H})$  and  one of the following holds,
\begin{enumerate}
\item $\beta \neq 0$ and one of  statements \ref{it:1} to \ref{it:8} in \crref{cor:GeoComInf} is true.
\item $\beta = 0$ and one of statements \ref{it:1}, \ref{it:3}, \ref{it:4}, \ref{it:7}, \ref{it:8} in \crref{cor:GeoComInf} is true.
\end{enumerate}
\end{theorem}

\begin{remark}
To see why we need to treat the case $d=2$  separately, note that any orthonormal basis $\{L_j\}$ would satisfy the conditions of the first part of the theorem with $\alpha=1$ and $\beta=0$. But this basis usually cannot be written in the form as specified in the theorem.
This is a consequence of the special features of the two-dimensional state space, such as the fact that \emph{every} Hermitian operator is a linear combination of a rank-$1$ projector and the identity, usually in two different ways.
\end{remark}
\begin{proof}
See the appendix.
\end{proof}

\section{\label{sec:AdjointUnitary}Group theoretic formulations of the SIC existence problem}
In this section we present two, closely related, group theoretic formulations of the SIC existence problem.  Group theory is a very rich and well-studied subject.  There are therefore grounds for hoping  that our results will make the SIC existence problem more tractable.  Our results may also be found  interesting from a group theoretic perspective.

 Let $\mrm{O}(d^2)$ be the group of  orthogonal $d^2\times d^2$ matrices.   We say that $R\in \mrm{O}(d^2)$  is of stochastic type if it is of the form
\begin{equation}
R = (d+1) S - d P,
\end{equation}
where $S$ is a doubly stochastic matrix \cite{MarsOA11book,HornJ85book},
and $P$ is the rank-$1$ projector defined in \eref{eq:AllOneMatrix}.
Note that $SP=PS=P$.  We say that a subgroup $G\subseteq \mrm{O}(d^2)$ is stochastic if it consists of matrices of stochastic type.  Any stochastic subgroup of $\mrm{O}(d^2)$ is contained in the subgroup $\{R\in \mrm{O}(d^2): \sum_k R_{jk}=1\}$ (note that $\sum_j R_{jk}=1$ is implicit from the definition), which is isomorphic to $\mrm{O}(d^2-1)$.

 Let $\mrm{PU}(d)$ be the projective unitary group in dimension~$d$ (\emph{i.e.}\ the unitary group $\mrm{U}(d)$ \emph{modulo} its center). We then have the following  characterization of the SIC existence problem:
\begin{theorem} \label{thm:StochType}
A SIC exists in dimension~$d$ if and only if $\mrm{O}(d^2)$ contains a  stochastic subgroup isomorphic to  $\mrm{PU}(d)$.
\end{theorem}
\begin{remark}
This connection between orthogonal matrices and doubly stochastic matrices plays an important role in the study of the symmetry properties of maximal consistent sets \cite{FuchS13,ApplEF11,ApplFZ13M}.
\end{remark}

The proof of this theorem will be given below, after we have proved our second group theoretic formulation of the SIC existence problem. Let $L=\{L_j\}$ be an orthonormal basis for $\mathcal{B}(\mathcal{H})$ consisting of Hermitian operators.  For each $U \in \mrm{U}(d)$ let
\begin{equation}
U^L_{jk} = \tr(L_jUL_kU^{\dagger})
\end{equation}
be the adjoint representation matrix  of $U$ with respect to this basis.   Define
\begin{equation}\label{eq:MinimalEntry}
m(d,L):=\min_{U,j,k} \bigl(U^L_{jk}\bigr), \quad m(d) :=\max_{L} \bigl( m(d,L)\bigr),
\end{equation}
where the maximum in the second definition is taken over all orthonormal bases consisting of Hermitian operators.  Our second group theoretic formulation of the SIC existence problem may now be stated as follows:
\begin{theorem}\label{thm:SICunitary}
For all $d\ge 2$
\begin{equation}
m(d) \le -\frac{1}{d}.
\end{equation}
The inequality is saturated  if and only if a SIC exists in dimension~$d$.

The bases $L$ for which $m(d,L) = -1/d$ are precisely the ones of the form
\begin{equation}\label{eq:AdjointBasis}
L_j=a\Pi_j+b,  \qquad a = \epsilon \sqrt{\frac{d+1}{d}}, \quad b = -\frac{a}{d}\left(1 -\epsilon' \sqrt{ \frac{1}{d+1}}\right),
\end{equation}
where $\{\Pi_j\}$ is a SIC and $\epsilon$, $\epsilon'$ are fixed signs.
\end{theorem}
\begin{remark}
Although we are not aware of any previous study on $m(d)$, it seems to us that this result is potentially important.  The research  on SICs \cite{DelsGS75, ReneBSC04, Zaun11,Gras04,Appl05,Gras05,Gras06,Gras08S,Gras08C, ScotG10, ApplBBG12,ApplAZ13} shows  that this upper bound can be saturated  at least for dimensions~2 to~16, 19, 24, 28, 31, 35, 37, 43, 48 (and saturated with high precision for dimensions up to~67), and it suggests the conjecture that $m(d) = -1/d$ in every finite dimension. This is indeed a remarkable contribution of SIC study to representation theory. We believe that the interplay between the two subjects will lead to more progress.
\end{remark}

Before proving \thref{thm:SICunitary}, we need to introduce a technical lemma, which may be of independent interest.
\begin{lemma}\label{lem:inverseProdGen}
Let $\lambda=(\lambda_1,\lambda_2,\ldots,\lambda_d)$ be a vector in $\mathbb{R}^d$ with $d\ge 2$. Then
\begin{equation}
\lambda^{\uparrow}\cdot \lambda^{\downarrow}\leq\frac{r^2-s}{d-1},
\end{equation}
where $r=\sum_j \lambda_j$, $s = \sum_j \lambda_j^2$ and  $\lambda^{\uparrow}$ \textup{(}respectively $\lambda^{\downarrow}$\textup{)} is the vector with the components of $\lambda$ in increasing \textup{(}respectively decreasing\textup{)} order.  The inequality is saturated if and only if at least $d-1$ components of $\lambda$ are identical, that is,
\begin{equation}\label{eq:lambdaOptimal}
\lambda^\downarrow=r^\prime(1,0,\ldots,0)+x\quad
\text{or}\quad \lambda^\downarrow=\frac{r^\prime}{d}(2,2,\ldots,2,2-d)+x,
\end{equation}
where
\begin{equation}
r^\prime=\sqrt{\frac{ds-r^2}{d-1}},\quad x=\frac{(d-1)r-\sqrt{(d-1)(ds-r^2)}}{d^2-d},
\end{equation}
and where  $x$ is understood as $x(1,1,\ldots,1)$ when appearing in a vector equation.
\end{lemma}
\begin{proof}  See the appendix.   \end{proof}

\begin{proof}[Proof of \thref{thm:SICunitary}]  Let $\lambda_{j,k}$ denote the eigenvalues of $L_j$; then we have $\sum_k \lambda_{j,k}=\tr(L_j)$ and $\sum_k\lambda_{j,k}^2=\tr(L_j^2) =1$.
Therefore,
\begin{align} \label{eq:mdlChain}
m(d,L)&= \min_{U, j,k}\tr(L_j UL_k U^\dag)  =\min_{j,k}\lambda_j^\uparrow \cdot \lambda_k^\downarrow \leq\min_{j}\lambda_j^\uparrow \cdot \lambda_j^\downarrow \nonumber \\
&\leq \min_j \frac{[\tr(L_j)]^2-1}{d-1}
\leq \frac{1}{d^2}\sum_j \frac{[\tr(L_j)]^2-1}{d-1} =-\frac{1}{d},
\end{align}
which establishes the upper bound on $m(d,L)$.  Here the second equality follows from the well known fact (see page 341 of \rcite{MarsOA11book} for example) that $\tr(L_j U L_k U^{\dag}) \ge \lambda^{\uparrow}_j \cdot \lambda^{\downarrow}_k$ for all $U$, $j$, $k$, and the observation that with an appropriate choice of $U$ the operators $L_j$ and $UL_kU^{\dag}$ are simultaneously diagonalizable.  The second inequality follows from \lref{lem:inverseProdGen}.
In deriving the last equality, we have applied the formula $\sum_j [\tr(L_j)]^2=\tr(1^2)=d$, which follows from the assumption that $\{L_j\}$ is an orthonormal basis of Hermitian operators.

We next prove that if the bound on $m(d,L)$ is saturated then a SIC exists in dimension~$d$ and the basis is of the form specified by \eref{eq:AdjointBasis}.  To show this observe that if the third  inequality is saturated we must have  $|\tr(L_j)|=1/\sqrt{d}$  for all $j$, while it follows from \lref{lem:inverseProdGen} that if  the second inequality is saturated then each $L_j$ has at least $d-1$ identical eigenvalues.  Consequently, each $L_j$ is a linear combination of a rank-1 projector and the identity.
According to  \thref{thm:SIC} with $\alpha=1$ and $\beta=0$, there exists a SIC  $\{\Pi_j\}$ such that
\begin{equation}
L_j = a_j \Pi_j + b_j,
\end{equation}
where
\begin{align}
a_j &= \epsilon_j \sqrt{\frac{d+1}{d}}, & b_j &= -\frac{a_j}{d}\left(1-\epsilon'\sqrt{\frac{1}{d+1}}\right),
\end{align}
and  $\epsilon_j$, $\epsilon'$ are signs (application of \thref{thm:SIC} is a fast recipe for deriving this conclusion although it is also not difficult to verify this claim directly  in this simple situation).  Observe that with these values of $a_j$, $b_j$,
\begin{equation}
\lambda_j^{\uparrow} \cdot \lambda_k^{\downarrow}
=\begin{cases}
-\frac{1}{d} \qquad & \text{if $\epsilon_j = \epsilon_k$},\\
-1 \qquad & \text{if $\epsilon_j \neq \epsilon_k$}.
\end{cases}
\end{equation}
To saturate the first inequality in \eref{eq:mdlChain}, all the signs $\epsilon_j$ must  equal a fixed sign, $\epsilon$ say. \Eref{eq:AdjointBasis} now follows.

To prove sufficiency it is enough to observe that $\lambda^{\uparrow}_j\cdot \lambda^{\downarrow}_k = -1/d$ for any basis of the type specified by \eref{eq:AdjointBasis}.
\end{proof}

We now turn to  proving  \thref{thm:StochType}.  It depends on the following lemma:
\begin{lemma} \label{lm:AdjRepPud}
Up to equivalence the adjoint representation of $\mrm{PU}(d)$ for $d\geq2$ is the only non-trivial irreducible representation of $\mrm{PU}(d)$ with degree not larger than $d^2-1$.
\end{lemma}
\begin{proof}
See the appendix.
\end{proof}

\begin{proof}[Proof of \thref{thm:StochType}]
To prove necessity  let $L$ be an Hermitian orthonormal basis of the type specified by \eref{eq:AdjointBasis}.  We have
$\tr(L_j) =  \epsilon\epsilon'/\sqrt{d}$ and $\sum_{j} L_j =  \epsilon\epsilon' \sqrt{d}$.
So
\begin{equation}
\sum_j U^L_{jk} = \sum_j U^L_{kj} = 1
\end{equation}
for all $U\in\mrm{U}(d)$.  It follows from this and \thref{thm:SICunitary} that if we define
\begin{equation}
S_{jk}= \frac{1}{d+1} \left( U^L_{jk} + \frac{1}{d}\right),
\end{equation}
then $S$ is doubly stochastic. Therefore,
$U^L = (d+1) S  - d P$ is
of stochastic type, and  the group $G = \{U^L \colon U \in \mrm{U}(d)\}$ is stochastic.  Meanwhile, $G$ is isomorphic to $\mrm{PU}(d)$ since the kernel of the homomorphism $U \mapsto U^L$ is the center of $\mrm{U}(d)$.

To prove sufficiency let $G$ be a stochastic subgroup of $\mathrm{O}(d^2)$ isomorphic to $\mrm{PU}(d)$. Then any isomorphism from $\mathrm{PU}(d)$ to $G$ defines a nontrivial representation of $\mathrm{PU}(d)$ of degree at most $d^2-1$, recall that any stochastic subgroup of $\mathrm{O}(d^2)$ is also a subgroup of $\mathrm{O}(d^2-1)$.
 It follows from \lref{lm:AdjRepPud}  that there exists an Hermitian orthonormal basis $L$ for $\mathcal{B}(\mathcal{H})$ such that $G=\{U^L \colon U\in \mrm{U}(d)\}$.    The claim then follows from  \thref{thm:SICunitary}, together with the fact  that entries of matrices of stochastic type are bounded from below by $-1/d$.
\end{proof}

It is interesting to ask what are the maximal stochastic subgroups of $\mrm{O}(d^2)$ (where by a "maximal stochastic subgroup" we mean a stochastic subgroup not properly contained in any larger stochastic subgroup).
The proof of \thref{thm:StochType} shows that if a SIC exists in dimension~$d$ then the adjoint representation of $\mrm{U}(d)$ relative to the basis $L$ defined by \eref{eq:AdjointBasis}
is a stochastic subgroup of $\mrm{O}(d^2)$.  However, it is not maximal stochastic.  On the other hand, we do get a maximal stochastic subgroup if we consider the adjoint representation of
the extended unitary group $\mrm{EU}(d)$ (\emph{i.e.}\ the group of all unitary and anti-unitary operators in dimension~$d$).
\begin{theorem}
Suppose a SIC exists in dimension~$d$, and  $L$ is a basis of the type specified by \eref{eq:AdjointBasis}.  Let $G\subseteq \mrm{O}(d^2)$ be the subgroup consisting of the adjoint representatives of $\mrm{EU}(d)$ relative to $L$.  Then $G$ is maximal stochastic.
\end{theorem}
\begin{proof}
The fact that $G$ is stochastic follows from a variant of the argument used in \thsref{thm:SICunitary} and~\ref{thm:StochType} to prove that the adjoint representation of $\mrm{U}(d)$ relative to $L$ is stochastic, note that conjugation by anti-unitary operators does not change the spectrum of Hermitian operators.

To prove that $G$ is maximal stochastic let
 $H\subseteq \mrm{O}(d^2)$ be a stochastic subgroup containing $G$, and let $R\in H$ be arbitrary.  Define a linear transformation $f\colon \mathcal{B}(\mathcal{H}) \to \mathcal{B}(\mathcal{H})$ by
\begin{equation}
f(L_j) = \sum_{k} R_{kj} L_k.
\end{equation}
Then $\tr(f(P)f(P^\prime))=\tr(PP^\prime)$ for any two pure states $P$ and $P^\prime$.
We will  show that $f$ takes pure states to  pure states.  It will then follow from Wigner's theorem  \cite{Wign59} that $f$ must be conjugation by a unitary or anti-unitary operator,  thereby implying that $R\in G$.

Let $P$ be an arbitrary pure state.  We begin by showing that $f(P)$ is positive semidefinite, that is,
$\tr (P' f(P)) \ge 0$
for every pure state $P'$.  Choose unitary operators $U$, $V$ such that
\begin{align}
P &=V \Pi_1 V^{\dag} , & P' &= U \Pi_1 U^{\dag}.
\end{align}
We have
\begin{equation}\label{eq:fmap}
\tr\left( L_1 U^{\dag} f(V L_1 V^{\dag})U\right) = \left( \left(U^{L} \right)^\rmT R V^{L}\right)_{11} \ge -\frac{1}{d}.
\end{equation}
The facts that $\tr(L_j) =  \epsilon\epsilon'/\sqrt{d}$, $\sum_j L_j = \epsilon\epsilon' \sqrt{d}$, and  $\sum_j R_{jk} =1$ mean that $f$ is unital and trace preserving, namely,  $f(1) =1$ and $\tr(f(A))=\tr(A)$ for any Hermitian operator $A$.
In view of \eref{eq:AdjointBasis}, we can now deduce from  \eref{eq:fmap} the following equation,
\begin{equation}
\tr (P' f(P))=\tr\left( \Pi_1 U^{\dag} f(V \Pi_1 V^{\dag})U\right) \ge 0.
\end{equation}
In addition,   we also have $\tr(f(P))=\tr\bigl( \left(f(P)\right)^2\bigr) = 1$ since  $R$ is orthogonal and $f$ is trace preserving.  Therefore, $f(P)$ is  a pure state, as claimed.
\end{proof}

The problem of determining all the maximal stochastic subgroups of $\mrm{O}(d^2)$ is, in general, difficult.  When $d=2$, however,  it has a  simple solution.
\begin{theorem}
When $d=2$ there is exactly one maximal stochastic subgroup of $\mrm{O}(d^2)$, namely, the group
\begin{equation}
G = \{A\in \mrm{O}(d^2)\colon \sum_k A_{jk} = 1\}.
\end{equation}
\end{theorem}
\begin{remark}
This theorem is closely related to the fact that in dimension 2 there is a unique maximal consistent set, namely, the quantum state space~\cite{ApplEF11, ApplFZ13M}.
\end{remark}

\begin{proof}
Observe that the set $G$ is indeed a group and that every matrix of stochastic type is contained in $G$.  So the result will follow if $G$ is  stochastic. To see this let $A$ be any element in $G$ and  $(x,y,z,w)$ any row in $A$.  We have
\begin{align}
x+y+z+w & =1, & x^2+y^2+z^2 + w^2 =1.
\end{align}
Using the method of Lagrange multipliers, it is straightforward to verify  that the minimum of $x$ under the two constraints is equal to   $-1/2$ and is attained when $ y=z=w=1/2$. The same analysis shows  that $A_{jk} \ge -1/2$ for all $j$ and $k$.  Consequently, $A$ is of stochastic type and the group $G$ is stochastic.
\end{proof}

\section{\label{sec:LieAlgebra}A Lie algebraic formulation of the SIC existence problem}
In the last section we established a connection between the SIC existence problem and the adjoint representation of the Lie \emph{group} $\mrm{U}(d)$.  In this section we establish an ostensibly quite different connection with the adjoint representation of the Lie \emph{algebra} $\mathfrak{u}(d)$ (\emph{i.e.}\ the Lie algebra of $\mrm{U}(d)$).    The result we prove is a much stronger version of a result previously proved by Appleby, Flammia and Fuchs~\cite{ApplFF11}.

Let $L = \{L_j\}$ be a basis for $\mathfrak{u}(d)$.  We adopt the physicist's convention that $\mathfrak{u}(d)$ consists of Hermitian matrices (as opposed to anti-Hermitian ones).  So the $L_j$ are all Hermitian.  Let $C^{L}_{jkl}$ be the structure constants for this basis:
\begin{equation}\label{eq:LieStrucDef}
[L_j,L_k] = \sum_{l} C^{L}_{jkl} L_l.
\end{equation}
We also define the structure matrices $C^L_j$ to be the matrices with elements
\begin{equation}\label{eq:LieStrucMatDef}
(C^L_j)_{kl} = C^L_{jkl}.
\end{equation}
Note that the structure constants and structure matrices are pure imaginary (as can be seen by taking Hermitian conjugates on both sides of \eref{eq:LieStrucDef}).

The significance of the structure matrices is that they are the adjoint representatives of the basis elements.  Thus, if $\ad_{A}$ is the linear map $\mathfrak{u}(d) \to \mathfrak{u}(d)$ defined by
\begin{equation}
\ad_{A} (B) = [A,B],
\end{equation}
and $C^L_A$ is the matrix defined by
\begin{equation}
\ad_{A} (L_k) =\sum_l (C^L_A)_{kl} L_l,
\end{equation}
then $C^L_j=C^L_{L_j}$.

In passing, we mention a well-known connection between the spectrum of $C^L_A$ and that of $A$, which will be needed in the sequel.
If $A$ has spectrum $\{\lambda_j, 1\leq j\leq d\}$, then $C^L_A$ has spectrum $\{\lambda_j-\lambda_k, 1\leq j,k\leq d \}$ (see, for example, Lemma 11 of \rcite{ApplFF11}). In particular, $C^L_A$ has rank $2(d-1)$ if and only if $A$ is a linear combination of a rank-1 projector and the identity and is not proportional to the identity.

We are now ready to state our main result:
\begin{theorem}\label{thm:LieSIC}
If $d\geq3$, then the following statements are equivalent:
\begin{enumerate}
\item A SIC exists in dimension~$d$.
\item There exists a basis for $\mathfrak{u}(d)$ such that the structure matrices are Hermitian and rank $2(d-1)$.
\end{enumerate}
The bases required in statement~2 are precisely the ones of the form
\begin{equation}\label{eq:SICAlgBasis}
L_j = \epsilon_j \ell(\Pi_j+\eta)
\end{equation}
where $\{\Pi_j\}$ is a SIC, the $\epsilon_j$ are signs, $\ell$ is a non-zero real number, and $\eta$ is an arbitrary real number not equal to $-1/d$.
\end{theorem}
\begin{remark}
In \rcite{ApplFF11}, in addition to the requirement that the structure matrices are Hermitian and rank $2(d-1)$, it was also required that they have the specific form $Q^{\vphantom{T}}_j - Q^\rmT _j$, where  the $Q_j$ are rank-$(d-1)$ projectors which are orthogonal to their own transposes.  The present theorem does not impose the last requirement and therefore represents a considerable strengthening.  In fact, this property will come for free once the weaker requirement in our theorem is satisfied. Note, however, that it only holds for $d>2$.  When $d=2$ it is necessary to fall back on the theorem proved in \rcite{ApplFF11}. The reason that the theorem does not hold when $d=2$ is that the proof depends on the first part of \thref{thm:SIC} .
\end{remark}

To prove \thref{thm:LieSIC} we  need the following lemma, which generalizes   Lemma 9  in \rcite{ApplFF11}.
\begin{lemma}\label{lem:structureconstants}
Let $L=\{L_j\}$ be a  basis for  $\mathfrak{u}(d)$.
Then the following statements are equivalent:
\begin{enumerate}
\item The structure matrices are Hermitian.
\item The structure constants are completely anti-symmetric.
\item The structure matrices  are of the form $C^L_j=H_j-H_j^\rmT$, where the $H_j$ are positive-semidefinite Hermitian matrices orthogonal to their own transposes.
\item The basis satisfies the three equivalent  \esref{eq:Ginner}, \eqref{eq:G2design} and \eqref{eq:GframeSuper}   of \thref{thm:GeoComInf}.
\end{enumerate}
If these conditions are satisfied, and $L_j$ is a rank-1 projector plus a multiple of the identity, then $H_j$ in statement~3 can be chosen to be a projector
of rank $d-1$.
\end{lemma}
\begin{proof}
As noted earlier the structure constants and structure matrices are pure imaginary.  So the structure matrices are Hermitian if and only if they are anti-symmetric. Since the structure constants are automatically anti-symmetric in the first two indices, the equivalence  $1 \equi 2$ follows.

The implication $3 \imply 1$ is immediate.  To prove the implication $1\imply 3$ observe that we can write the structure matrices in the form $C^L_j = H_j - N_j$, where $H_j$ and $N_j$ are orthogonal positive semidefinite Hermitian matrices.  The anti-symmetry of the structure matrices then implies that $N^{\vphantom{T}}_j = H^\rmT _j$. If, in addition, $L_j$ has spectrum $\{\lambda_{j,k}: 1\leq k\leq d\}$, then $H_j$ has nonzero spectrum $\{\lambda_{j,k}-\lambda_{j,l}: \lambda_{j,k}>\lambda_{j,l}\}$. In particular, $H_j$ is a rank-$(d-1)$ projector if $L_j$ is a rank-1 projector plus a multiple of the identity.

To prove the equivalence $2\equi  4$   we first show that statement $2$ is equivalent to the requirement that $\sum_j\douter{L_j}{L_j}$ commutes with $\ad_{L_k}$ for all $k$.  In fact
\begin{equation}
\dbra{L_j} \ad_{L_k} \dket{L_l} = \tr(L_j [L_k,L_l]) = \tr(L_l [L_j,L_k]) = \sum_m C^L_{jkm}\dinner{L_m}{L_l}.
\end{equation}
Consequently,
\begin{equation}
\dbra{L_j} \ad_{L_k} = \sum_m C^L_{jkm} \dbra{L_m}
\end{equation}
and
\begin{equation}
\left( \sum_j\douter{L_j}{L_j} \right) \ad_{L_k} = \sum_{m,j} C^L_{jkm} \douter{L_j}{L_m} = - \sum_{m,j} C^L_{kjm} \douter{L_j}{L_m}.
\end{equation}
On the other hand,
\begin{equation}
\ad_{L_k}\left( \sum_j\douter{L_j}{L_j} \right) = \sum_{m,j} C^L_{kjm} \douter{L_m}{L_j} = \sum_{m,j} C^L_{kmj} \douter{L_j}{L_m}.
\end{equation}
The claim now follows.  We next observe that the semi-simplicity of the Lie algebra $\mathfrak{su}(d)$ means that the adjoint representation of $\mathfrak{u}(d)$ has only two irreducible components, namely, $\mathfrak{su}(d)$ and the one-dimensional subspace spanned by the identity.   According to  Schur's lemma, the requirement that $\sum_j\douter{L_j}{L_j}$ commutes with $\ad_{L_k}$ for all $k$ is equivalent to \eref{eq:GframeSuper} of  \thref{thm:GeoComInf}.
\end{proof}

\begin{proof}[Proof of \thref{thm:LieSIC}] To prove the implication $1\imply 2$ let $\{\Pi_j\}$ be a SIC in dimension~$d$ and let $L=\{L_j\}$ be the basis specified by \eref{eq:SICAlgBasis}.  Then $L$ satisfies \eref{eq:Ginner}  of  \thref{thm:GeoComInf}. It  follows from \lref{lem:structureconstants}   that the structure matrices are Hermitian.
The structure matrices  also  have rank $2(d-1)$ since each $L_j$ is a linear combination of a rank-1 projector and the identity.

To prove the implication $2\imply 1$ observe that if statement $2$ holds then it follows from \lref{lem:structureconstants} that the basis satisfies the three \esref{eq:Ginner}, \eqref{eq:G2design} and \eqref{eq:GframeSuper}   of \thref{thm:GeoComInf}.   Meanwhile, the rank condition ensures that each  $L_j$ is a linear combination of a rank-1 projector and the identity. According to  \thref{thm:SIC}, there exists a SIC $\{\Pi_j\}$ such that
\begin{equation}
L_j = \epsilon_j\ell(\Pi_j+\eta),
\end{equation}
as specified by \eref{eq:SICAlgBasis}.
\end{proof}

By virtue of  \lref{lem:structureconstants}, it is now straightforward to prove the additional statement proved in \rcite{ApplFF11}, that SIC existence implies the existence of a basis $L$ for which the structure matrices have the form $C^L_j = Q_j - Q^\rmT _j$, where the $Q_j$ are rank-$(d-1)$ projectors which are orthogonal to their own transposes.  This nice property comes for free once there exists a basis such that the structure matrices are Hermitian and rank $2(d-1)$.
Details are left to the reader.

\Thref{thm:LieSIC} and \lref{lem:structureconstants} have several interesting consequences.
When the basis operators $L_j$ have rank one, the structure matrices $C^L_j$ automatically have rank $2(d-1)$. The following corollary is an immediate consequence of \thref{thm:LieSIC} and this observation when $d\geq3$, but it also holds when $d=2$.
\begin{corollary}\label{cor:RankOneBasis}
Suppose $\{L_j\}$  is a   basis for the Lie algebra $\mathfrak{u}(d)$ that is composed of rank-1 Hermitian operators. Then the structure constants with respect to this basis are completely anti-symmetric if and only if
$\epsilon_j L_j/\ell $ with $\epsilon_j:=\sgn(\tr(L_j))$ form a SIC for some positive constant  $\ell$.
\end{corollary}
Alternatively, this corollary can be derived  as follows. According to \lref{lem:structureconstants}, the sufficiency is immediate. Conversely, if the structure constants are completely anti-symmetric, then $\sum_j L_j\otimes L_j=(\beta+\alpha)P_{\mrm{s}}+  (\beta-\alpha)P_{\mrm{a}}$. Since $L_j$ have rank one, $L_j\otimes L_j$ is supported on the symmetric subspace. It follows that $\beta=\alpha$ and $\sum_j L_j\otimes L_j=2\alpha P_{\mrm{s}}$. Define  $\ell=\sqrt{\alpha(d+1)/d}$ and $\Pi_j=\epsilon_j L_j/\ell$.
Then $\Pi_j$ are $d^2$ pure states (the normalization is not assumed here but will follow later) that satisfy $\sum_j \Pi_j\otimes \Pi_j=[2d/(d+1)] P_{\mrm{s}}$. According to \thref{thm:minimal2design}, $\{\Pi_j\}$ is a SIC and the corollary follows.

\begin{corollary}\label{cor:PureBasis}
Suppose $\{\Pi_j\}$ is  a  basis for the Lie algebra $\mathfrak{u}(d)$ that is composed of pure states. Then the structure constants  are completely anti-symmetric if and only if   $\{\Pi_j\}$ is a SIC.
\end{corollary}

In the case of a SIC the structure constants are completely anti-symmetric.  It is interesting to ask for what other bases this is true.  The following two theorems provide a partial answer to that question.
\begin{theorem}\label{thm:NormalizedBasis}
Suppose $\{L_j\}$ is a   basis for the Lie algebra $\mathfrak{u}(d)$ for which $\tr(L_j^2)$ is constant. Let $\epsilon_j=1$ if $\tr(L_j)\ge 0$ and $-1$ otherwise. Then the following statements are equivalent:
\begin{enumerate}
\item The structure constants are completely anti-symmetric.
\item $\{\epsilon_j L_j\}$ is a regular simplex and, in case the $\epsilon_j L_j$ are not orthogonal,  $|\tr(L_j)|$ is a non-zero constant.
\end{enumerate}
\end{theorem}
\begin{remark}
Unlike the usual definition of  sigh factors, here $\epsilon_j$ are nonzero even if $\tr(L_j)=0$.
\end{remark}

\begin{proof}
To prove the implication $1\imply 2$ observe that, according to \lref{lem:structureconstants},  if the  structure constants are completely anti-symmetric, then
\begin{equation}
\tr(L_j L_k)=\alpha\delta_{jk}+\frac{\beta}{\alpha+d\beta}\tr( L_j)\tr(L_k)
\end{equation}
 with $\alpha, \alpha+d\beta>0$.
 If $\beta=0$, then  $\{\epsilon_j L_j\}$ is an orthonormal basis and thus forms a regular simplex. Otherwise,  the constancy of $\tr(L_j^2)$ implies the constancy of $|\tr(L_j)|$,  and $\{\epsilon_jL_j\}$  again forms
  a regular simplex. The fact that the constant value of $|\tr(L_j)|$ is non-zero follows from the fact that $\{L_j\}$ is a basis.

To prove the implication $2\imply 1$ observe that statement $2$ implies
\begin{equation}
\tr(\epsilon_j L_j\epsilon_k L_k)=\alpha\delta_{jk}+\gamma\tr(\epsilon_j L_j)\tr(\epsilon_kL_k)
\end{equation}
for some real constants $\alpha, \gamma$.   The fact that the structure constants are completely anti-symmetric then follows  from \lref{lem:structureconstants}.
\end{proof}

\begin{theorem}
Suppose $\{L_j\}$ is a  basis for the Lie algebra $\mathfrak{u}(d)$ for which $|\tr(L_j)|$ is constant.  Let $\epsilon_j=1$ if $\tr(L_j)\ge 0$ and $-1$ otherwise. Then the following statements are equivalent:
\begin{enumerate}
\item The structure constants are completely anti-symmetric.
\item $\{\epsilon_j L_j\}$ is a regular simplex.
\end{enumerate}
\end{theorem}
\begin{proof}
The argument is almost the same as in the proof of  \thref{thm:NormalizedBasis}.
\end{proof}

\section{\label{sec:JordanAlgebra}A Jordan algebraic formulation of the SIC existence problem}
In the last section we established a connection between the SIC existence problem and the Lie algebra $\mathfrak{u}(d)$.  In this section we establish a connection which is in some ways analogous between the SIC existence problem and the Jordan algebra $\mathfrak{j}(d)$: \emph{i.e.}\ the algebra consisting of the set of operators on $\mathcal{H}$ equipped with   the anti-commutator
\begin{equation}
\{A,B\}=AB+BA
\end{equation}
as product (note that in the literature it is more common to take the product to be the anti-commutator scaled by the factor $1/2$).
We will also have occasion to consider the Jordan algebra $\mathfrak{j}_{\mrm{H}}(d)$  of all Hermitian operators on $\mathcal{H}$ and equipped with the same product.  The connections between Jordan algebras and homogeneous self-dual cones \cite{Wilc09,Wilc11,BarnW12} mean that this formulation of the SIC problem may be relevant to convex-operational approaches to quantum mechanics.

Let $L=\{L_j\}$ be a basis for $\mathfrak{j}(d)$ consisting of Hermitian operators (so $L$ is also a basis for $\mathfrak{j}_{\mrm{H}}(d)$).  Analogously to \esref{eq:LieStrucDef} and \eqref{eq:LieStrucMatDef} we define the structure constants $C^L_{jkl}$ by
\begin{equation}\label{eq:JordanStrucDef}
\{L_j,L_k\} = \sum_l C^L_{jkl} L_l
\end{equation}
and the structure matrices  $C^L_j$ by
\begin{equation}
(C^L_j)_{kl} = C^L_{jkl}.
\end{equation}
Note that for a given Hermitian basis the structure constants and structure matrices are the same irrespective of whether we consider the algebra $\mathfrak{j}(d)$ or the algebra $\mathfrak{j}_{\mrm{H}}(d)$.  Note also that if the basis is Hermitian the structure constants and structure matrices are real (as can be seen by taking Hermitian conjugates on both sides of \eref{eq:JordanStrucDef}).

Just as with the Lie algebraic formulation we need to make a distinction between the cases $d=2$ and $d>2$.  We accordingly prove two theorems:  a stronger one (\thref{thm:JordanSIC}) which holds when $d>2$ and which is analogous to \thref{thm:LieSIC}  in this paper, and a weaker one  (\thref{thm:JordanSIC2}) which holds when $d\ge 2$ and which is analogous to Theorem~7 of Appleby, Flammia and Fuchs \cite{ApplFF11}. The weaker theorem, though entails more assumptions, is still interesting (it establishes an analogy with the $Q-Q^\rmT $ property in the Lie algebraic case).
\begin{theorem}\label{thm:JordanSIC}
If $d>2$ the following statements are equivalent:
\begin{enumerate}
\item A SIC exists in dimension~$d$.
\item There exists a basis for $\mathfrak{j}_{\mrm{H}}(d)$ such that  each structure matrix is  a rank-$(2d-1)$ symmetric matrix  plus a multiple of the identity.
\end{enumerate}
The bases satisfying statement $2$ are precisely the ones of the form
\begin{align}\label{eq:JordAlgMain1LDef}
L_j = \epsilon\epsilon_j c(\Pi_j-a), \quad a= \frac{d+1-\epsilon\sqrt{d+1}}{d(d+1)},
\end{align}
where $\{\Pi_j\}$ is a SIC, $\epsilon$, $\epsilon_j$ are signs, and $c$ is a positive constant.
\end{theorem}
\begin{remark}
Thus, in the Lie algebraic case the structure matrices are required to be anti-symmetric, while in the Jordan algebraic case they are required to be symmetric.  However, in the later case the rank condition is slightly more complicated.

Observe that requiring the structure matrices to be symmetric is equivalent to requiring the structure constants to be completely symmetric---by contrast to the Lie algebraic case where the structure constants are required to be completely anti-symmetric.
\end{remark}

\begin{theorem}\label{thm:JordanSIC2}
If $d\ge 2$ the following statements are equivalent:
\begin{enumerate}
\item A SIC exists in dimension~$d$.
\item There exists a basis $L$ for $\mathfrak{j}_{\mrm{H}}(d)$ such that each structure matrix is of the form
\begin{equation}
C^L_j= Q^{\vphantom{T}}_j+Q_j^\rmT + 2P^{\vphantom{T}}_j -2a_j,
\end{equation}
where $Q_j$ is  a  rank-$(d-1)$ projector which is orthogonal to its own transpose, $P_j$  a real rank-$1$ projector orthogonal to $Q^{\vphantom{T}}_j$ and $Q^\rmT _j$, and $a_j$  a real constant.
\end{enumerate}
If these statements hold, the constants $a_j$ in statement $2$ are all equal to $a$ in \eref{eq:JordAlgMain1LDef},
and the basis $L$ is given by
\begin{equation}\label{eq:JordanSICBasis}
L_j = \Pi_j-a,
\end{equation}
where $\{\Pi_j\}$ is a SIC.
\end{theorem}
\begin{remark}
In the Lie algebraic case each structure matrix is of the form $Q^{\vphantom{T}}_j -Q^\rmT _j$ where $Q_j$ is a rank $d-1$ projector orthogonal to its own transpose.  The Jordan algebraic case is similar to that, but slightly more complicated.

Note that the  $P_j$, being real, are automatically symmetric---so the result is consistent with \thref{thm:JordanSIC}.
The $Q_j$ are necessarily not real (a real projector is  identical to its transpose, and so cannot be orthogonal to it except in the trivial case of the zero projector).  However, the combination $Q^{\vphantom{T}}_j + Q^\rmT _j$ is, of course, real.

Interestingly, the basis in \eref{eq:JordanSICBasis} is identical with the one in \eref{eq:AdjointBasis} apart from an overall scale factor.
\end{remark}

Before proving these theorems we need to develop some machinery.  By analogy with the adjoint representation of a Lie algebra there is, associated to each element $A\in\mathfrak{j}(d)$,  a linear map $f_A \colon \mathfrak{j}(d) \to \mathfrak{j}(d)$ defined by
\begin{equation}
f_A(B) = \{A,B\}.
\end{equation}
Note that, unlike the Lie algebraic case, this does not give us a representation of the algebra since it is  generally not the case that $f_{\{A,B\}} = \{f_A,f_B\}$. If $A$ is Hermitian we make no notational distinction between the map $f_A \colon \mathfrak{j}(d) \to \mathfrak{j}(d)$, and its restriction to the space of Hermitian operators $f_A \colon \mathfrak{j}_{\mrm{H}}(d) \to \mathfrak{j}_{\mrm{H}}(d)$.

It is often more convenient to represent linear maps such as $f_A$ with matrices.  Given a basis $L=\{L_j\}$ for $\mathfrak{j}_{\mrm{H}}(d)$,   the anticommutator $\{A, L_k\}$ can be expanded as
\begin{equation}\label{eq:JordanRep}
\{A,L_k\}=\sum_l (C^L_A)_{kl}L_l.
\end{equation}
Then the transpose of the matrix $C^L_A$ defined by the expansion coefficients is the matrix representation of $f_A$ relative to the basis $L$. Note also that $C^L_{L_j}$ is just the structure matrix~$C^L_j$.

Let $\{\lambda_j \colon j = 1,\dots, d\}$ be the eigenvalues of $A\in \mathfrak{j}_{\mrm{H}}(d)$.  It is straightforward to verify, analogously to the Lie algebraic case (\emph{cf.}\ the proof of Lemma~11 in \rcite{ApplFF11}), that the eigenvalues of $f_A$ and that of $C^L_A$ are $\{\lambda_j + \lambda_k\colon j,k=1,\dots, d\}$ (note that the eigenvalues of $f_A$ are the same, irrespective of whether one considers it as acting on $\mathfrak{j}_{\mrm{H}}(d)$ or $\mathfrak{j}(d)$).    In particular, if $A$ is a real constant, say $A=c$, then $f_A$ is $2c$ times the identity map and $C^L_A=2c$ irrespective of the specific basis.
We  also have the following analogue of Lemma~11 of \rcite{ApplFF11}:
  \begin{lemma}\label{lem:JordanRank}
Suppose $A$ is a Hermitian operator and $L=\{L_j\}$ is any basis for $\mathfrak{j}_{\mrm{H}}(d)$.
If $d\neq 3$, then  $f_A$ and $C^L_A$ have rank $2d-1$ if and only if $A$ has rank 1.
If $d=3$, then $f_A$ and $C^L_A$ have rank $2d-1$ if and only if $A$ has rank 1 or  has spectrum of the form $\{\lambda, -\lambda,-\lambda\}$ with $\lambda\neq0$.
\end{lemma}

\begin{proof}
Though a little more tedious,  the proof proceeds along essentially the same lines as the proof of Lemma~11 of \rcite{ApplFF11}.
\end{proof}

In the Lie algebraic case the fact that the adjoint representation of $\mathfrak{u}(d)$ has only  two irreducible components played an important role (\emph{cf.}\ the proof of \lref{lem:structureconstants}).  In the Jordan algebraic case the situation is even simpler since there is only one irreducible component, as the following Lemma shows.

\begin{lemma}\label{lem:JordanIrreducible}
The action of the superoperators $\{f_A\colon A \in \mathfrak{j}_{\mrm{H}}(d)\}$ is irreducible on both $\mathfrak{j}(d)$ and~$\mathfrak{j}_{\mrm{H}}(d)$.
\end{lemma}
\begin{proof}
Suppose, on the contrary,  that the action on $\mathfrak{j}(d)$ was reducible. Let $\mathcal{S}$ be a nontrivial  invariant subspace under this action and $\mathcal{S}^\bot$ its orthogonal complement with respect to the Hilbert-Schmidt inner product. Then $\mathcal{S}^\bot$ is also invariant according to the following equation with $B\in \mathcal{S}^\bot$ and $C\in \mathcal{S}$:
\begin{equation}
\dinner{C}{f_A(B)}=\tr(C^\dag AB)+\tr(C^\dag BA)=\dinner{f_{A^\dag}(C)}{B}=0.
\end{equation}
We now establish a contradiction by considering a special operator.
Let $\{|r\rangle\colon r=1, \dots, d\}$ be an orthonormal basis for $\mathcal{H}$ and let $E_{rs} = |r\rangle \langle s|$.   Define $A=\sum_r \lambda_r E_{rr}$ with $\lambda_r = 10^r$.  Then $E_{rr}$ is the unique eigenvector of $f_A$ with eigenvalue $2\lambda_r$.  Consequently, each $E_{rr}$ belongs either to $\mathcal{S}$ or to $\mathcal{S}^\bot$.  The equation
\begin{equation}
f_{E_{rs}}(E_{rr})=E_{rs}=f_{E_{rs}} (E_{ss}),\quad s\neq r,
\end{equation}
then implies that  $\{E_{rs} \colon r,s=1,\dots,d\}$ is contained   either in $\mathcal{S}$ or in $\mathcal{S}^{\bot}$.  But the $E_{rs}$ are a basis for $\mathfrak{j}(d)$ so this contradicts the assumption that $\mathcal{S}$ is non-trivial.

The irreducibility of  the action on $\mathfrak{j}_{\mrm{H}}(d)$  is  an easy consequence of its irreducibility on~$\mathfrak{j}(d)$.
\end{proof}

Before proving our main results,  \thsref{thm:JordanSIC} and~\ref{thm:JordanSIC2}, we need to establish the following lemma.
\begin{lemma}\label{lem:JordanStructure}
Let $L=\{L_j\}$ be a basis for $\mathfrak{j}_{\mrm{H}}(d)$.  Then the following statements are equivalent:
\begin{enumerate}
\item $C^L_A$ is   symmetric for any Hermitian operator $A$.

\item For any Hermitian operator $A$, $C^L_A$ can be written as  $C^L_A=2S+H+H^\rmT$, where  $H$ is a Hermitian matrix which is orthogonal to its own transpose, and where $S$ is a real symmetric matrix  orthogonal to $H$ and  $H^\rmT$ and having the same nonzero spectrum as~$A$.

\item $\{L_j/\ell\}$ is an orthonormal basis for some positive constant $\ell$.
\end{enumerate}
If, in addition, $A$ is a rank-1 projector, then $H$ and $S$ in statement~2 can be chosen to be projectors with ranks $d-1$ and 1, respectively.
\end{lemma}
\begin{remark}
Note that the first statement is equivalent to the statement that the structure matrices are symmetric, which in turn is equivalent to the statement that the structure constants are completely symmetric.  The second statement is equivalent to the requirement that the structure matrices $C^L_j = C^L_{L_j}$ have the stated form. Here "orthogonal" means having orthogonal support. The lemma may thus be regarded as an analogue of \lref{lem:structureconstants} for the Lie algebraic case.
\end{remark}
\begin{proof}
The implication $2\imply1$ is immediate. To prove the equivalence $1\equi 3$,
note that statement 1 is equivalent to the statement that $C^L_j$ are symmetric for all $j$ since $A$ is a linear combination of the basis elements $L_j$.
Define
\begin{equation}
D_{jkl}=\tr(\{L_j,L_k\}L_l),\quad (D_j)_{kl}=D_{jkl},\quad M_{jk}=\tr(L_jL_k).
\end{equation}
 Then $D_j$ and $M$ are real symmetric matrices satisfying
$D_j=C^L_jM=M(C^L_j)^\rmT $. In addition,  $M$ is positive definite since $\{L_j\}$ is a basis.
Therefore, $M$ commutes with $C^L_j$  if and only if $C^L_j$ is symmetric.
It follows from  \lref{lem:JordanIrreducible} and an analog of Schur's lemma that $M$ commutes with all the $C^L_j$ if and only if it is proportional to the identity.   Consequently,   the $C^L_j$ are all symmetric if and only if  $\{L_j/\ell\}$ is an orthonormal basis for some positive constant $\ell$. The equivalence $1\equi 3$ follows.

It remains to prove the implication $3\imply 2$.  When $\{L_j/\ell\}$ is an orthonormal basis,
\begin{equation}\label{eq:OrthonormalCaseCA}
(C^L_A)_{kl}=\frac{1}{\ell^2}\tr(L_l f_A(L_k)),
\end{equation}
from which it follows immediately that $C^L_A$ is a real symmetric matrix.
Suppose $A$ has the spectral decomposition $A=\sum_{r} \lambda_r \outer{e_r}{e_r}$. Define  $E_{rs} = |e_r\rangle \langle e_s|$. Then  \eref{eq:OrthonormalCaseCA}  reads in superoperator notation
\begin{equation}
(C^L_A)_{kl} = \frac{1}{\ell^2} \sum_{r,s} (\lambda_r+\lambda_s) \dinner{L_l}{E_{rs}} \dinner{E_{rs}}{L_k}.
\end{equation}
Define
\begin{align}
S_{kl}&= \frac{1}{\ell^2} \sum_{r} \lambda_r \dinner{L_l}{E_{rr}} \dinner{E_{rr}}{L_k},\\
H_{kl} &= \frac{1}{\ell^2} \sum_{r<s} (\lambda_r+\lambda_s) \dinner{L_l}{E_{rs}} \dinner{E_{rs}}{L_k}.
\\
\intertext{We have}
(H^\rmT )_{kl} &= \frac{1}{\ell^2} \sum_{r>s} (\lambda_r+\lambda_s) \dinner{L_l}{E_{rs}} \dinner{E_{rs}}{L_k},
\end{align}
and consequently,
\begin{equation}
C^L_A = 2 S^{\vphantom{T}}+ H^{\vphantom{T}} + H^\rmT .
\end{equation}
Observe that $S$ and $H$ are the transposes of the matrix representations of the superoperators $\sum_{r} \lambda_r \douter{E_{rr}}{E_{rr}}$ and  $\sum_{r<s} (\lambda_r+\lambda_s) \douter{E_{rs}}{E_{rs}}$  with respect to the orthonormal basis $\{L_j/\ell \}$. We conclude that $S$ is real symmetric and $H$ Hermitian; $S$, $H$ and $H^\rmT$ are mutually orthogonal. Moreover, $S$ and $H$ have nonzero spectrum
\begin{equation}
\{\lambda_r: \lambda_r\neq0, \;1\leq r\leq d\},  \quad \{\lambda_r+\lambda_s: \lambda_r+\lambda_s\neq0,\; 1\leq r<s\leq d\},
\end{equation}
respectively. In particular, $S$ has the same nonzero spectrum as $A$.

If, in addition, $A$ is a rank-1 projector, then $H$ and $S$ defined above are  projectors with ranks $d-1$ and 1, respectively, which completes the proof.
\end{proof}

We are now ready to prove our main results.
\begin{proof}[Proof of \thref{thm:JordanSIC}]
To prove the implication $1\imply 2$, let $\{\Pi_j\}$ be a SIC and  $L_j$ the basis defined  by \eref{eq:JordAlgMain1LDef}. Then $L_j/\ell$ with $\ell=c\sqrt{d/(d+1)}$ form an orthonormal basis. So the $C^L_j$  are real symmetric according to \lref{lem:JordanStructure}.  Since each $L_j$ is a linear combination of a rank-1 projector and the identity, each $C^L_j$ is a linear combination of a real symmetric matrix of rank $2d-1$ and the identity by \lref{lem:JordanRank}.

To prove the implication $2 \imply 1$, let $\{L_j\}$ be a basis that satisfies the requirement in statement~2. Then it follows from  \lref{lem:JordanStructure} that $\{L_j/\ell\}$ is an orthonormal basis for some positive constant $\ell$.
In addition, there exist  real constants $a_j$ such that each $C_{L_j+a_j}$ has rank $2d-1$. According to \lref{lem:JordanRank}, $L_j+a_j$ has rank 1  or, in the case $d=3$,  has spectrum of the form $\{\lambda_j, -\lambda_j, -\lambda_j\}$. In any case, each $L_j$ is a linear combination of a rank-1 projector and the identity.
According to   \thref{thm:SIC} with $\alpha=\ell^2$ and $\beta=0$, there exists a SIC $\{\Pi_j\}$ such that  $L_j$ have the form specified in \eref{eq:JordAlgMain1LDef} (note that this  last step is not valid when $d=2$, which is why  the theorem only holds for $d>2$).
\end{proof}

\begin{proof}[Proof of \thref{thm:JordanSIC2}]  To prove the implication $1\imply 2$ let $\{\Pi_j\}$ be a SIC and  $\{L_j\}$ the basis  in \eref{eq:JordanSICBasis}. Then  $\{L_j/\ell\}$ is an orthonormal basis with $\ell = \sqrt{d/(d+1)}$. Define $\tilde{C}^L_j$ as the transpose of the matrix representation of $f_{\Pi_j}$ with respect to the basis $\{L_j\}$; then $C^L_j=(\tilde{C}^L_j-2a)$. By \lref{lem:JordanStructure} and the fact that $\Pi_j$ is a rank-1 projector, we find $\tilde{C}^L_j=Q_j+Q_j^\rmT+2P_j$, where $P_j$ and $Q_j$ satisfy the requirement of statement~2 in the theorem. The implication $1\imply 2$ follows.

It remains to prove the implication $2\imply 1$. When $d>2$,  according to \thref{thm:JordanSIC}, there exists a SIC $\{\Pi_j\}$ such that the basis satisfying statement~2 has the form $L_j=c\epsilon \epsilon_j (\Pi_j-a)$,
where $c$, $\epsilon$, $\epsilon_j$ and $a$ are as specified in \thref{thm:JordanSIC}. By essentially the same argument that leads to the implication $1\imply 2$ we find that the $C^L_j$ can be written  as
\begin{equation}
C^L_j=c\epsilon \epsilon_j (Q_j^\prime+{Q_j^\prime}^\rmT + 2 P_j^\prime-2a),
\end{equation}
 where $P_j^\prime$ and $Q_j^\prime$ have the same properties as $P_j$ and $Q_j$. Inspection of the spectrum of $C^L_j$ shows that this equality and the assumption $C^L_j= (Q_j+Q_j^\rmT+ 2 P_j -2a_j)$ can be satisfied simultaneously if and only $\epsilon \epsilon_jc =1$ and $a_j=a$. Therefore, $L_j$ have the form specified in \eref{eq:JordanSICBasis}.

When  $d=2$, the symmetry of the structure matrices  implies that $C^L_A$ is symmetric for all Hermitian operators $A$, and consequently that $\{L_j/\ell\}$ is an orthonormal basis for some positive constant $\ell$ by \lref{lem:JordanStructure}.  So  the $L_j$ satisfy the three equivalent  \esref{eq:Ginner}, \eqref{eq:G2design} and~\eqref{eq:GframeSuper}   of \thref{thm:GeoComInf} with $\alpha=\ell^2$ and $\beta=0$.    Let $\lambda_{j,1}\geq \lambda_{j,2}$ be the eigenvalues of  $L_j$.  Given that  $C^L_j$ is the transpose of the matrix form of $f_{L_j}$ relative to the basis $L$, the eigenvalues of $C^L_j$ are $2\lambda_{j,1}, \lambda_{j,1}+\lambda_{j,2}, \lambda_{j,1}+\lambda_{j,2}, 2\lambda_{j,2}$ in nonincreasing order.   From the assumption $C^L_j= (2P_j+Q_j+Q_j^\rmT-2a_j)$ we deduce
 \begin{align}
 \lambda_{j,1} &= 1-a_j , & \lambda_{j,2} &= -a_j,
  \end{align}
which implies  $2\tr(L_j^2)-[\tr(L_j)]^2=( \lambda_{j,1}- \lambda_{j,2})^2=1$, so that statement 4 in \crref{cor:GeoComInf} holds. According to \thref{thm:SIC} with $\alpha=\ell^2$ and $\beta=0$,  there exists a SIC $\{\Pi_j\}$ such that $L_j=c\epsilon \epsilon_j (\Pi_j-a)$,
where   $\epsilon$, $\epsilon_j$ are signs, $c$ is a real constant, and $a$ is given by \eref{eq:JordAlgMain1LDef}.   By the same argument as in the case $d>2$,  we find $c\epsilon \epsilon_j =1$ and $a_j=a$. Again, $L_j$ have the form specified in \eref{eq:JordanSICBasis}.
\end{proof}

\section{\label{sec:summary}Summary}
We have explored various   group theoretic and algebraic characterizations of the SIC existence problem based on a unified framework. In particular, we proved the equivalence of the following statements:
\begin{enumerate}
\item[1.\hphantom{a}] The existence of a SIC in dimension~$d$.

\item[2a.] The existence of a stochastic subgroup of $\mrm{O}(d^2)$ that is isomorphic to the projective unitary group $\mrm{PU}(d).$

\item[2b.] The existence of an adjoint matrix representation  of the unitary group $\mrm{U}(d)$ such that all matrix elements are bounded from below by $-1/d$.

\item[3.\hphantom{a}] The existence of a basis for the Lie algebra $\mathfrak{u}(d)$ such that each  structure matrix is Hermitian with rank $2(d-1)$.

\item[4.\hphantom{a}]  The existence of a  basis for the Jordan algebra $\mathfrak{j}_{\mrm{H}}(d)$ such that each structure matrix is a  linear combination of a rank-$(2d-1)$ real symmetric matrix and the identity matrix.

\end{enumerate}
In conjunction with well-known geometric, combinatoric, and information theoretic characterizations, these new characterizations not only enrich the  meanings and implications of SICs, but also point to new directions for attacking the SIC existence problem.   Besides, our discovery may prove to be valuable to studying  the unitary group, Lie algebra, and Jordan algebra.

Our study further demonstrates that the SIC existence problem is not an isolated problem, not just a geometric curiosity, but has deep consequences, which are pertinent to a wide range of research fields. We hope our work will stimulate more interest and progress on this topic.

\section*{Acknowledgements}
We thank Ingemar Bengtsson, Lin Chen, and Gelo Noel Tabia for discussions and Blake Stacey for comments. H.Z. also thanks Dragomir \v{Z} {\DJ}okovi\'{c} for discussions and  Aakumadula for answering a question on the representations of the projective unitary group posed on  MathOverflow. This research was supported in part by Perimeter Institute for Theoretical Physics. Research at Perimeter Institute is supported by the Government of Canada through Industry Canada and by the Province of Ontario through the Ministry of Research and Innovation. D.M.A was supported by the IARPA MQCO program, by the ARC via EQuS project number CE11001013, by the US Army Research Office grant numbers W911NF-14-1-0098 and W911NF-14-1-0103, by the U.S. Office of Naval Research (Grant No. N00014-09-1-0247), and by the John Templeton Foundation.

\appendix

\numberwithin{equation}{section}

\section{Technical details}
\subsection{Proof of \crref{cor:GeoComInf}}
\begin{proof}
When $\beta\neq 0$, the equivalence of statements \ref{it:1} to  \ref{it:5}  is an immediate consequence of \esref{eq:Ginner}, \eqref{eq:LjSumA} and \eqref{eq:LjSumB}, note that $\gamma<1/d$ according to \eref{eq:gamma}.

The implication $\ref{it:5}\imply \ref{it:6}$ is trivial. Conversely, if statement~\ref{it:6} holds, then
\begin{equation}
\frac{\gamma^2[\tr(L_j)]^2[\tr(L_k)]^2}{\{\alpha+\gamma[\tr(L_j)]^2\} \{\alpha+\gamma[\tr(L_k)]^2\} }=\frac{\gamma^2[\tr(L_j)]^2[\tr(L_m)]^2}{\{\alpha+\gamma[\tr(L_j)]^2\} \{\alpha+\gamma[\tr(L_m)]^2\}  }
\end{equation}
whenever $j,k,m$ are distinct. Therefore,
\begin{equation}
\frac{[\tr(L_k)]^2}{ \{\alpha+\gamma[\tr(L_k)]^2\} }=\frac{[\tr(L_m)]^2}{ \{\alpha+\gamma[\tr(L_m)]^2\} }
\end{equation}
for all $k,m$, which implies that $|\tr (L_j)|$ is independent of $j$. As a consequence, statements \ref{it:1} to \ref{it:6} are equivalent.

The implication $\ref{it:1}\imply  \ref{it:7}$ follows from \eref{eq:LjLj}.  To show that $\ref{it:7}  \imply  \ref{it:8}$ observe that if statement~\ref{it:7} holds, then \eref{eq:Ginner} implies that
\begin{equation}
\tr(\epsilon_jL_j\epsilon_kL_k)=\alpha\delta_{jk}+\gamma\tr(\epsilon_jL_j)\tr(\epsilon_kL_k).
 \end{equation}
Summing over $j,k$ and letting $A=\sum_j \epsilon_jL_j$ yields
 \begin{equation}
 \tr(A^2)=d^2\alpha+\gamma [\tr(A)]^2.
 \end{equation}
 Suppose $A=\kappa$; then $\kappa>0$ and $d\kappa^2=d^2\alpha+d^2\gamma\kappa^2$, which implies
\begin{equation}
\kappa=\sqrt{\frac{d\alpha}{1-d\gamma}}=\sqrt{d(\alpha+d\beta)}.
\end{equation}
As a consequence, $\sum_j |\tr(L_j)|=\sum_j \epsilon_j\tr(L_j)=d\kappa=d\sqrt{d\alpha+d^2\beta}$, thereby verifying the implication $\ref{it:7}\imply  \ref{it:8}$.

Finally to show that $\ref{it:8} \imply  \ref{it:1}$, note that
\begin{equation}
 \sum_j [\tr(L_j)]^2=d\alpha+d^2\beta
\end{equation}
according to \eref{eq:LjSumA}. So $\sum_j |\tr(L_j)|\leq d\sqrt{d\alpha+d^2\beta}$, and the upper bound is saturated if and only if $|\tr(L_j)|=\sqrt{d\alpha+d^2\beta}/d$ for all $j$, that is, the value of $|\tr(L_j)|$ is independent of~$j$.

The first two equations in \eref{eq:Lj} follow from \esref{eq:LjSumA}  and \eqref{eq:LjSumB} together with statements \ref{it:1} and \ref{it:2}; the third equation  follows from statements \ref{it:7} and \ref{it:8}.

When $\beta=0$ statement  \ref{it:2} is automatic.   The equivalence of statements \ref{it:1}, \ref{it:3} and \ref{it:4} is immediate.
 The equivalence of statements \ref{it:1}, \ref{it:7} and \ref{it:8} follows the same reasoning as in the case $\beta\neq0$.  The equivalence of statements \ref{it:5} and \ref{it:6} is also immediate.  If  the value of $|\tr(L_j)|$ is independent of $j$, then $L_j$ cannot be traceless since they form a basis. Therefore, statements \ref{it:1}, \ref{it:3}, \ref{it:4}, \ref{it:7} and \ref{it:8} imply  statements \ref{it:5} and \ref{it:6}.
\end{proof}

\subsection{Proof of \crref{cor:RegularSimplexGen}}

\begin{proof}
The Gram matrix $\tr(L_j L_k)$ has two distinct eigenvalues $\alpha$ and $\alpha+ d^2\zeta$.  The fact that  $\{L_j\}$ is a basis means that the Gram matrix must be positive definite, implying  $\alpha>0$ and $\zeta > - \alpha /d^2$.

When  $\sum_j L_j $ is proportional to the identity, say $\sum_j L_j =\eta $,  summing over $k$ in \eref{eq:RegularSimplex} yields $\eta\tr(L_j)=\alpha+d^2\zeta$. Therefore, $\eta\neq 0$ and $\tr(L_j)=(\alpha+d^2\zeta)/\eta$ is independent of $j$.

If the value of $\tr(L_j)$ is independent of $j$ and is equal to $\ell$, then $\ell$ cannot be zero since $\{L_j\}$ is a basis. Now the  equation
\begin{equation}
\tr\biggl[\biggl(\sum_j L_j\biggr)L_k\biggr]=\sum_j(\alpha\delta_{jk}+\zeta)=\alpha+d^2\zeta=\frac{\alpha+d^2\zeta}{\ell}\tr(L_k)
\end{equation}
implies that $\sum_j L_j=(\alpha+d^2\zeta)/\ell$ is proportional to the identity. Taking the trace of this equation yields $d^2\ell=d(\alpha+d^2\zeta)/\ell$. Therefore,
$|\ell |=\sqrt{(d\alpha+d^3\zeta)}/d$, from which statement~3 and \eref{eq:RegularSimplexLj} follow immediately.

To show the implication $3\imply  1$, define $A=\sum_j L_j$, then \eref{eq:RegularSimplex} implies the equality $\tr(A^2)=d^2\alpha+d^4\zeta$.
So $|\tr(A)|\leq d\sqrt{d\alpha+d^3\zeta}$ and the inequality is saturated if and only if $A$ is proportional to the identity.

Next, suppose in addition $\zeta\neq 0$. If any of  statements 1, 2, 3 holds, then $\tr(L_j)$ is nonzero and independent of $j$, so  \eref{eq:RegularSimplex} implies  \eref{eq:Ginner} with $\gamma=\zeta/[\tr(L_j)]^2=d\zeta/(\alpha+d^2\zeta)$. Conversely, if  \eref{eq:Ginner} holds with $\gamma=d\zeta/(\alpha+d^2\zeta)$, then \eref{eq:RegularSimplex} implies that the value of $\tr(L_j)$ is independent of $j$,   from which statements 1, 2, 3 follow.
\end{proof}

\subsection{Proof of \thref{thm:SIC}}
\begin{proof}Suppose  $L_j = a_j \Pi_j + b_j$ with $\Pi_j$ rank-1 projectors and $a_j, b_j$ real constants.
We have
\begin{align}
&(\beta+\alpha) P_{\mrm{s}}+(\beta-\alpha) P_{\mrm{a}}=\sum_j L_j\otimes L_j\nonumber\\
&=\sum_{j} a_j^2\Pi_j\otimes \Pi_j +A\otimes 1 +1 \otimes  A+\sum_j b_j^2 1 \otimes 1,
\end{align}
where $A=\sum_j a_j b_j \Pi_j$. As a consequence,
\begin{equation}
(\beta-\alpha)P_\mrm{a}=P_\mrm{a}\biggl(\sum_j L_j\otimes L_j\biggr)P_\mrm{a}=P_\mrm{a}(A\otimes 1 +1 \otimes  A)P_\mrm{a}+\sum_j b_j^2 P_\mrm{a}.
\end{equation}
When $d\geq3$, it is not hard to show that this equality holds  if and only if $A=(\beta-\alpha-\sum_j b_j^2)/2$ (the conclusion is not valid when $d=2$ since  the range of $P_\mrm{a}$ only has dimension one, which is why the case $d=2$ has to be handled separately). In that case, $\sum_j a_j^2\Pi_j\otimes \Pi_j=2\alpha P_{\mrm{s}}$,
so  $a_j^2=\alpha(d+1)/d$ and $\{\Pi_j\}$ is a SIC according to \thref{thm:minimal2design}. In particular, $\{\Pi_j\}$ is a basis for $\mathcal{B}(\mathcal{H})$ and $\sum_j\Pi_j=d $.
Now the equality
\begin{equation}
\sum_j a_j b_j\Pi_j=A=\frac{1}{2}\biggl(\beta-\alpha-\sum_j b_j^2\biggr)
\end{equation}
implies that the values of $a_jb_j$ and $b_j^2$ are independent of $j$, which further implies that
\begin{equation}
da_jb_j=\frac{1}{2}(\beta-\alpha-d^2b_j^2),
\end{equation}
\Eref{eq:ajbj} now follows given that $a_j^2 = \alpha(d+1)d$. The theorem holds when $d\geq3$.

In the case  $d=2$, if  $\{L_j\}$ is a basis and any of the two additional assumptions detailed in the theorem holds, then the same conclusion as in the case $d>3$ holds. Since decompositions $L_j = a_j \Pi_j + b_j$ are generally not unique when $d=2$, however, the $\Pi_j$ introduced at the beginning of the proof do not necessarily form a SIC.  To resolve this problem, we need to take a slightly different approach so that these decompositions are chosen consistently.

According to  \thref{thm:GeoComInf} we have
\begin{align}\label{eq:LjLkd2}
\tr(L_jL_k)=\alpha\delta_{jk}+\frac{\beta}{\alpha+2\beta}\tr(L_j)\tr(L_k).
\end{align}
According to  \crref{cor:GeoComInf},
\begin{equation}
\tr(\epsilon_jL_j)=\sqrt{\frac{\alpha+2\beta}{2}},\quad \tr(L_j^2)=\frac{2\alpha+\beta}{2},
\end{equation}
and  $\{\epsilon_j L_j\}$ forms a regular simplex ($\epsilon_j$ being the sign of $\tr(L_j)$).
The two eigenvalues of $\epsilon_jL_j$ are given by
\begin{equation}
\lambda^{\pm}=\frac{1}{2}\biggl(\sqrt{\frac{\alpha+2\beta}{2}}\pm\sqrt{\frac{3\alpha}{2}}\biggr),
\end{equation}
along with  the eigen-projectors
\begin{equation}\label{eq:Pijd2}
\Pi_j^{\pm}=\pm\frac{\epsilon_jL_j-\lambda^{\mp}}{\lambda^+-\lambda^-}.
\end{equation}
Since $\{\epsilon_j L_j\}$ forms a regular simplex, $\{\Pi_j^+\}$ and $\{\Pi_j^-\}$ form two sets of equiangular lines and thus two SICs according to \thref{thm:MaximalEAL} (as  can also be verified by using   \eref{eq:LjLkd2}). Consequently,
\begin{equation}
L_j=\epsilon_j\biggl[\sqrt{\frac{3\alpha}{2}} \Pi_j^{+} -\frac{1}{2}\biggl(\sqrt{\frac{3\alpha}{2}}-\sqrt{\frac{\alpha+2\beta}{2}}\biggr)\biggr]
=-\epsilon_j\biggl[\sqrt{\frac{3\alpha}{2}} \Pi_j^{-} -\frac{1}{2}\biggl(\sqrt{\frac{3\alpha}{2}}+\sqrt{\frac{\alpha+2\beta}{2}}\biggr)\biggr].
\end{equation}
Both  decompositions of $L_j$ have the form specified in the theorem, which completes the proof.
\end{proof}

\subsection{Proof of \lref{lem:inverseProdGen}}
\begin{proof}
The claim is trivial when $d=2$. When $d\geq3$ we first prove the result for the special case $r=s=1$, and then use that  to prove it for the general case.

Permuting the components of $\lambda$ does not change the value of $\lambda^{\uparrow}\cdot \lambda^{\downarrow}$.  There is therefore no loss of generality in assuming that the components of $\lambda$ are arranged in decreasing order to begin with, so that $\lambda^{\downarrow} = \lambda$. Let $\sigma$ be any permutation of the integers $1$ to $d$ and let $\lambda^{\sigma}$ be the vector with components $\lambda^{\sigma}_j = \lambda_{\sigma(j)}$.  We claim that
\begin{equation} \label{eq:lambdaIneq}
\lambda\cdot \lambda^{\sigma} \ge \lambda \cdot \lambda^{\uparrow} = \lambda^{\downarrow}\cdot \lambda^{\uparrow}.
\end{equation}
To see this observe that we can bring $\lambda^{\sigma}$ into  increasing order by  the following iterative  procedure.  Transpose any pair of adjacent components $\lambda^{\sigma}_j, \lambda^{\sigma}_{j+1}$ for which $\lambda^{\sigma}_j > \lambda^{\sigma}_{j+1}$ and  denote by $\lambda^{\sigma'}$ the vector which results. We have
\begin{equation}
\lambda \cdot \lambda^{\sigma} - \lambda\cdot \lambda^{\sigma'} = (\lambda_j - \lambda_{j+1})(\lambda^{\sigma}_{j} - \lambda^{\sigma}_{j+1}) \ge 0.
\end{equation}
\Eref{eq:lambdaIneq} then follows from successive applications of the above procedure.

Now let $P_{\mrm{f}}$ be the set of free permutations of the integers $1$ to $d$ (\emph{i.e.}\ the set of permutations having no fixed points).  We have
\begin{equation} \label{eq:freePermId}
\sum_{\sigma \in P_{\mrm{f}}} \lambda \cdot \lambda^{\sigma}= c \sum_{i\neq j} \lambda_i \lambda_j=c(r^2-s)=0,
\end{equation}
where $c$ is a positive integer whose specific value is irrelevant to us. Here the last equality follows from the assumption $r=s=1$.

Inequality \eqref{eq:lambdaIneq} and \eref{eq:freePermId} together imply $
\lambda^{\uparrow}\cdot\lambda^{\downarrow} \le 0$,
which establishes the first statement of the lemma for the case $r=s=1$.  To establish the second statement observe that if the inequality is saturated we must have $
\lambda \cdot \lambda^{\sigma} =0$
for all $\sigma \in P_{\mrm{f}}$.  Let $\sigma$ be the permutation $(1,\dots, d)$ and  $\sigma'$  the permutation $(1,d)(2,\dots, d-1)$ (where, as usual, $(j_1, \dots , j_n)$ denotes the cyclic permutation $j_1 \to j_2 \to \dots \to j_n \to j_1$).  Then
\begin{equation}
(\lambda_1-\lambda_{d-1})(\lambda_2-\lambda_d)  = \lambda \cdot \lambda^{\sigma} - \lambda \cdot \lambda^{\sigma'} = 0.
\end{equation}
So either $\lambda_1 =\lambda_{d-1}$ or $\lambda_2 = \lambda_d$.  Taking into account of the fact that $\lambda$ is in decreasing order  we deduce that $d-1$ components of $\lambda$ are identical.  Together with the assumption $\sum_j\lambda_j=\sum_j\lambda_j^2=1$, this observation implies that
\begin{align}
\lambda^{\downarrow} &= \lambda =(1,0,\dots, 0) \quad \text{or} \quad  \left( \frac{2}{d}, \dots , \frac{2}{d} ,\frac{2}{d} -1\right),
\end{align}
which is identical with \eref{eq:lambdaOptimal} in the case $r=s=1$.
Conversely, it is straightforward to verify the equality $\lambda^{\uparrow}\cdot\lambda^{\downarrow} = 0$
when  $d-1$ components of $\lambda$ are identical.

Let us now relax the assumption that $r=s=1$.  If $r^2 = ds$ it follows from the Cauchy inequality as applied to the vectors $\lambda$ and $(1,\dots, 1)$ that $\lambda = (r, \dots, r)/d$. The claim is then immediate.
 Otherwise, define
$\lambda'_j: = \eta \lambda_j + \xi$,
where
\begin{align}
\eta &= \sqrt{\frac{d-1}{ds-r^2}},  & \xi &= \frac{1-r\eta}{d}.
\end{align}
Then we have $\sum_j \lambda'_j =  \sum_j {\lambda'}^2_j = 1$.
The claim is now a direct consequence  of the result already proved.
\end{proof}

\subsection[Proof  of \lref{lm:AdjRepPud}]{Proof  of \lref{lm:AdjRepPud}\footnote{We are grateful to Aakumadula for suggesting this proof, in answer to a question posed by H.Z. on MathOverflow.}}
\begin{proof}
Since $\mrm{PU}(d) = \mrm{PSU}(d)$ (where $\mrm{PSU}(d)$ is $\mrm{SU}(d)$ modulo its  center),  we can focus on the representations of $\mrm{SU}(d)$.
The irreducible representations of $\mrm{SU}(d)$ are labeled by partitions  of the form $\lambda=[\lambda_1,\lambda_2,\ldots,\lambda_d]$ where $\lambda_j$ are integers such that $\lambda_1 \ge \dots \ge \lambda_d = 0$ (see, for example,  \rcite{FultH91}).  The irreducible representations of $\mrm{PSU}(d)$  are precisely those representations of $\mrm{SU}(d)$ which are trivial on the center, namely, the representations for which $\sum_j \lambda_j = 0 \mod d$.  The dimension of the representation labeled by $\lambda$ is denoted by $D_{\lambda}$ and given by the Weyl dimension formula
\begin{equation}\label{eq:WeylDForm}
D_\lambda=\prod_{1\leq j<k\leq d}\frac{\lambda_j-\lambda_k+k-j}{k-j}.
\end{equation}
We will prove the lemma by explicitly enumerating the representations for which $D_{\lambda}\le d^2-1$ and showing that, aside from the trivial representation, the adjoint representation is the only one for which $\sum_j \lambda_j = 0 \mod d$.

In the following it will be convenient to identify the partition $[\lambda_1,\dots, \lambda_d]$ with $[\lambda_1, \dots, \lambda_r]$ when $\lambda_{r+1} = \dots = \lambda_d=0$ (since $\lambda_d=0$ this means we can always write $[\lambda_1,\dots,\lambda_{d-1}]$ in place of  $[\lambda_1,\dots,\lambda_d]$).   Denote by $\lambda^{\mrm{C}}$ the  conjugate partition, whose Young diagram is obtained from the Young diagram of $\lambda$ by interchanging the rows and columns.  Denote by $\lambda^{\mrm{D}}$ the dual, or contragredient partition $[\lambda_1-\lambda_d, \lambda_1-\lambda_{d-1},\dots, \lambda_1-\lambda_2,0]$.  Then  $D_{\lambda^{\mrm{D}}}=D_{\lambda^{\vphantom{\mrm{D}}}}$ for all $\lambda$ according to \eref{eq:WeylDForm}.

We begin by considering  partitions of the form $[m]^{\mathrm{C}}$, with $d>m>0$, for which the  Young diagrams have  one column. None of these partitions satisfy the requirement $\sum_j \lambda_j = 0 \mod d$.  However, for later use we need to consider for which  values of $m$ the dimension is bounded by $d^2-1$.  Using \eref{eq:WeylDForm} we find
\begin{equation}
D_{[m]^{\mathrm{C}}}=\binom{d}{m}.
\end{equation}
When $d\leq 7$, the inequality $D_{[m]^{\mathrm{C}}}\leq d^2-1$ is always satisfied; when $d=8$, it is satisfied as long as $m\neq 4$; when $d\geq9$, it is satisfied  if and only if $m$ takes on one of the four values $1,2, d-2,d-1$.

We next consider  partitions of the form $[m_1,m_2]^{\mathrm{C}}$, with $d>m_1 \ge m_2 > 0$, for which the Young diagrams have two columns.  Inspection of \eref{eq:WeylDForm} shows
\begin{align}
D_{[m_1]^{\mrm{C}}} &< D_{[m_1,m_2]^{\mrm{C}}},  & D_{[m_2]^{\mrm{C}}}  &< D_{[m_1,m_2]^{\mrm{C}}}.
\end{align}
When $d\ge 9$ we can use the results of the last paragraph to deduce that both $m_1$ and  $m_2$ are restricted to the four values $1,2, d-2,d-1$  if $D_{[m_1,m_2]^{\mrm{C}}} \le d^2-1$.  This  reduces the number of possible partitions to $10$.  When $d<9$ we cannot reduce the number of possibilities to be considered in this way, but since the search space is already small this does not matter.  A case-by-case examination reveals that the only two-column partitions for which $D_{\lambda}\leq d^2-1$ are $[1,1]^{\mathrm{C}}$, $[d-1,d-1]^{\mathrm{C}}$ and $[d-1,1]^{\mathrm{C}}$.  If $d>2$ then the only one of these three representations that satisfies the requirement $\sum_j \lambda_j = 0 \mod d$ is the adjoint representation $[d-1,1]^{\mathrm{C}}$.  The adjoint representation is also the only  one for which $D_{\lambda}=d^2-1$.  If $d=2$ the  three representations are identical, all coinciding with the adjoint representation, and satisfying the requirement $\sum_j \lambda_j = 0 \mod d$ and the equality $D_{\lambda}=d^2-1$.

We next consider three-column partitions $[m_1,m_2,m_3]^{\mrm{C}}$ with $d>m_1\ge m_2 \ge m_3 > 0$.  Inspection of \eref{eq:WeylDForm} shows
\begin{equation}
D_{[m_1,m_2]^{\mrm{C}}} < D_{[m_1,m_2,m_3]^{\mrm{C}}}.
\end{equation}
If $d=2$ the results of the last paragraph immediately imply that there are no three-column partitions for which $D_{\lambda}\le d^2-1$.  If $d>2$ then the only candidates are $[1,1,1]^{\mrm{C}}$ and $[d-1,d-1,m_3]^{\mrm{C}}$.  We find $D_{[1,1,1]^{\mrm{C}}} = d(d+1)(d+2)/6 > d^2-1$.  To exclude the other possibility observe
 that if the representation corresponding to $[d-1,d-1,m_3]^{\mrm{C}}$ satisfied the bound on the dimension, then so would the dual representation corresponding to
 $\left([d-1,d-1,m_3]^{\mrm{C}}\right)^{\mrm{D}} = [d-m_3,1,1]^{\mrm{C}}$.  Since $[d-m_3,1,1]^{\mrm{C}}$ is also a three-column partition, and since $D_{[1,1,1]^{\mrm{C}}} > d^2-1$, this would
 mean $[d-m_3,1,1]=[d-1,d-1,m'_3]$ for some $m'_3$.  But this contradicts the assumption $d>2$.  We conclude that, irrespective of the value of $d$, there are no
  three-column partitions for which $D_{\lambda}\leq d^2-1$. As a consequence, there are no
  $n$-column partitions with $n\geq3$ for which $D_{\lambda}\leq d^2-1$ since $D_{[m_1,m_2,m_3]^{\mrm{C}}} \leq D_{[m_1,\dots, m_n]^{\mrm{C}}}$ for any such partition $[m_1,\dots, m_n]^{\mrm{C}}$.
 \end{proof}

\bibliographystyle{ieeetr}

\bibliography{all_references}

\end{document}